\documentclass[format=sigconf]{acmart}

\usepackage{graphicx}

\usepackage{color}
\usepackage{url}
\usepackage{multirow}
\usepackage{rotating}
\usepackage[ruled,noend,linesnumbered]{algorithm2e}

\usepackage{subfigure}

\newcommand{\commentout}[1]{}

\usepackage[normalem]{ulem}
\newcommand{\revrev}[2]{#2}

\newcommand{\revtwo}[2]{#2}

\newcommand{\revthree}[2]{#2}

\newcommand{\revfour}[2]{#2}

\renewcommand\footnotetextcopyrightpermission[1]{}
\begin{document}

\title{Worst-case Bounds and Optimized Cache on $M^{th}$ Request Cache Insertion Policies under Elastic Conditions}

\author{Niklas Carlsson}
\affiliation{
  \institution{Link\"oping University, Sweden}
}
\email{niklas.carlsson@liu.se}

\author{Derek Eager}
\affiliation{
  \institution{University of Saskatchewan, Canada}
}
\email{eager@cs.usask.ca}

\begin{abstract}

Cloud services and other shared third-party infrastructures allow individual content providers to easily
scale their services based on current resource demands.  In this paper, we consider an individual content
provider that wants to minimize its delivery costs under the
\revrev{assumption}{assumptions}
that the
\revrev{}{storage and bandwidth}
resources it requires are
\revrev{elastic and}{elastic,}
the content provider only pays for the resources that it
\revrev{consumes.}{consumes, and costs are proportional to the resource usage.}
Within this context,
we (i) derive worst-case bounds for the optimal cost and competitive cost ratios of different classes of
{\em cache on $M^{th}$ request} cache insertion policies, (ii) derive explicit average cost expressions
and bounds under arbitrary inter-request distributions, (iii) derive explicit average cost expressions
and bounds for short-tailed
\revrev{(Erlang and deterministic), exponential,}{(deterministic, Erlang, and exponential)} and heavy-tailed (Pareto) inter-request distributions,
and (iv) present
\revrev{numerical evaluations and cost comparisons}{numeric and trace-based evaluations}
that
\revrev{reveals}{reveal}
insights into the
\revrev{policies relative costs performance.}{relative cost performance of the policies.}
Our results show that a window-based {\em cache on $2^{nd}$ request}
\revrev{}{policy}
using a single threshold optimized
to minimize worst-case costs provides good average performance across the different distributions and
the full parameter ranges of each considered distribution, making it an attractive choice for a wide
range of practical conditions where request rates of individual
\revfour{objects}{file objects}
typically are not known and
can change quickly.

\end{abstract}

\keywords{Caching, Worst-case bounds, Optimized insertion policies}

\maketitle

\fancypagestyle{firststyle}
               {
                 \fancyhf{}
                 \fancyfoot[CF]{\tiny
                   $^*$This is the authors' version of a work that was accepted for publication in {\em IFIP Performance}, Toulouse, France, Dec. 2018.
                   The final version will appear in
                   {\em Performance Evaluation}, volumes 127-128, Nov. 2018, pp. 70-92.\\
                   (\url{https://doi.org/10.1016/j.peva.2018.09.006})
                   Changes resulting from the publishing process, such as editing, corrections, structural formatting, and other
                   quality control mechanisms may not be reflected in this document.\\
                   Changes may have been made to this work
                   since it was submitted for publication.
                 }
               }
               \thispagestyle{firststyle}

\section{Introduction}

Cloud services and other shared infrastructures are becoming increasingly common.
These infrastructures are typically third-party operated and
allow individual service providers using them to easily scale their
services based on current resource demands.
In the context of content delivery,
rather than buying and operating their own dedicated servers,
many content providers are already using third-party operated
Content Distribution Networks (CDNs) and cloud-based content delivery platforms.
This trend towards using third-party providers on an on-demand basis is expected to
increase as new content providers enter the market.

Motivated by current on-demand cloud-pricing models, in this paper,
we consider an individual content provider that wants to minimize
its delivery costs under the
\revrev{assumption}{assumptions}
that
the resources it requires to deliver its service are
\revrev{{\em elastic} and}{{\em elastic},}
the content provider
\revrev{{\em only pays for the resources it consumes}.}{{\em only pays for the resources it consumes}, and {\em costs are proportional} to the resource usage.}
For the purpose of our analysis, we consider a simple cost model in which the content
provider pays the third-party service for (i) the amount of storage it consumes due to
caching close to the end-users and (ii) the amount of
\revrev{}{(backhaul)}
bandwidth that it and its end-users consume.
Under this model, we then analyze the optimized delivery costs of different
{\em cache on $M^{th}$ request} cache insertion policies when using a
\revfour{TTL-based}{Time-to-Live (TTL) based}
\revrev{caching system}{eviction policy}
in which
\revrev{objects remain in the cache after insertion
for as long as they are requested within a time interval $T$
of the most recent prior request to that same object.}{\revfour{an object}{a file object}
  remains in the cache after insertion
  until a time interval $T$ has elapsed without any requests for the object.}

\revfour{}{It is important to note that although use of a TTL eviction policy has been shown useful
in approximating the performance of a fixed-size Least-Recently-Used (LRU) cache
when the number of file objects is sufficiently
large~\cite{ChTW02, FrRR12, BDC+13, BGSC14, GaLM16,CaEa18b},
and our results may therefore provide some insight  for this case, it is not the focus of this paper.
Here we assume elastic resources, where cache eviction is not needed to make space for a new
insertion, but rather to reduce cost by removing objects that are not expected to be requested
again soon.  A TTL-based eviction policy is a good heuristic for such purposes.
Cloud service providers already provide elastic provisioning at varying granularities
for computation and storage, and in the context of trends such as serverless computing,
in-memory caching, and multi-access edge computing, we believe that support for fine-grained
elasticity may increase in the future.} 

In the past,
\revfour{\revfour{these discriminatory}{discriminatory}}{selective}
cache insertion policies have been shown
valuable in reducing cache pollution due to ephemeral content popularity and
the long tail of one-timers observed in edge networks~\cite{GALM07,ZSGK09, MaSi15,CaEa17}.
However, prior work has not bounded or optimized the worst-case delivery costs of such policies. 

In this paper, we first present novel worst-case bounds for the optimal cost and
competitive cost-ratios of different variations of these policies.
Second, we derive explicit average cost
\revfour{bounds and cost expressions}{expressions and cost ratio bounds}
for these policies
under arbitrary inter-request time distributions,
assuming independent and identically distributed request times,
as well as for
\revrev{(i) short-tailed inter-request time distributions (Erlang and deterministic),
  (ii) exponential inter-requests, and (iii) heavy-tailed inter-request time distributions (Pareto).}{specific short-tailed
  (deterministic, Erlang, and exponential) and heavy-tailed (Pareto) inter-request time distributions.}
\revtwo{}{Our analysis includes comparisons against both
  \revfour{optimal offline}{{\em optimal offline}}
  policy bounds
  and, for the case when hazard rates are increasing or constant,
  \revfour{optimal online}{{\em optimal online}}
  policy bounds; all derived here.}
Finally, we present
\revrev{numerical}{numeric and trace-based}
evaluations and provide insights
into the
\revrev{policies relative cost performance.}{relative cost performance of the policies.}

Our analysis reveals that window-based {\em cache on $M^{th}$ request}
cache insertion policies can substantially outperform policies that
\revrev{simply count the number of requests before the time of insertion.}{do not take into account the recency of
prior object requests when making cache insertion decisions.}
With window-based {\em cache on $M^{th}$ request} policies
\revrev{(i) a counter is maintained for how many times each object have been requested
within a time window $W$ of the most recent prior request for that object,
(ii) the counter is reset to one each time there was no request within the last time window $W$, and
(iii) the object is cached whenever the counter reaches $M$.}{a counter is
  \revfour{maintained for how many times each uncached object has been
  requested within a time window $W$ of the most recent request for that object (including the most recent request itself).
  The object is inserted into the cache whenever the counter reaches $M$.}{maintained for each uncached object
    that has been requested at least once within the last $W$ time units.
    A newly allocated counter is initialized to one, and the counter is
    incremented by one whenever the object is referenced within $W$ time units of its most recent previous request.
    The object is inserted into the cache whenever the counter reaches $M$.}}
Our results
\revfour{shows}{show}
that a single parameter version of this policy can be used beneficially,
in which $W=T$, and that the best worst-case bounds are achieved by selecting
the window size
\revfour{$W$}{$W=T$}
equal to the time that it takes to accumulate a cache storage cost (for that object)
equal to the remote bandwidth cost $R$ associated with a cache miss (for that object).
With these protocol settings,
the worst-case bounds of the window-based {\em cache on $M^{th}$ request} policies
have a competitive ratio of $M+1$
(compared to the {\em optimal offline} policy).
While these ratios at first may appear discouraging for larger $M$,
our average case analysis for different inter-request time distributions clearly shows substantial cost
benefits of using intermediate $M$ such as 2-4,
with the best choice depending on where in the parameter region the system operates.
For less popular objects a slightly larger $M$ (e.g., $M=4$) may be beneficial;
however, in general, window-based {\em cache on $M^{th}$ request} with $M=2$
typically provides the most consistently good average performance across the
full parameter ranges of each considered distribution.
Overall, the results show that using this policy with optimal worst-case parameter setting (i.e., $W=T=R$)
may be
attractive for
practical conditions,
where request rates of individual objects typically are not known and can change quickly.

The remainder of the paper is organized as follows.
Sections~\ref{sec:model} and~\ref{sec:policies}
present our system model and the practical
insertion policies considered, respectively.
Section~\ref{sec:worse-case} presents the
\revfour{optimal offline}{{\em optimal offline}}
policy and derives worst-case bounds for the different insertion policies.
Section~\ref{sec:general} presents cost expressions
for the
\revfour{optimal offline}{{\em optimal offline}}
bound under both arbitrary and specific distributions.
Section~\ref{sec:baselin} presents the corresponding expressions for
\revtwo{optimized baseline policies that assume}{an optimized baseline policy that assumes}
knowledge of the
\revfour{exact inter-request distributions and the current request intensity,}{precise inter-request time distribution for each object,}
and shows that this policy has the same performance
  as the
  \revfour{optimal online}{{\em optimal online}}
  policy when hazard rates are increasing or constant.
Section~\ref{sec:general-policies}
then derives general cost expressions for the practical
insertion policies, before Section~\ref{sec:results}
presents the distribution-specific expressions,
\revrev{analyzes the policies relative performance,}{analyzes the relative performance of the policies,}
and compares
\revrev{how their costs compare}{their costs}
against the
\revfour{offline optimal}{{\em offline optimal}}
and
\revrev{}{optimized}
baselines.
\revrev{}{Section~\ref{sec:multi} complements the single-file analysis results
  with both analytic and trace-based multi-file evaluations.}
Finally, Section~\ref{sec:related} discusses related work and
Section~\ref{sec:conclusion} presents our conclusions.

\section{System Model}\label{sec:model}

\revrev{Let}{Initially, let}
us consider the costs associated with
a single file object as seen at a single cache
location.
\revrev{}{(The multi-file case is considered in Section~\ref{sec:multi}.)}
Furthermore, without loss of generality,
for this object and location,
let us assume that the provider
\revrev{must pay}{pays}
(i) a (normalized) {\em storage cost} of $1$ per time unit that the file object
is stored in the cache and (ii) a {\em remote bandwidth cost} $R$ each time a request is made
to an object currently not in cache.
At these times, the
\revfour{file}{file object}
needs to be retrieved from the origin servers
\revrev{(or a different cache) and therefore result in}{(or a different cache), which results in}
additional bandwidth costs (and delivery delays).
Note that $R$ is defined as the incremental delivery cost, beyond that of delivering
the content from the cache to the client.  This latter
\revrev{}{(typically much smaller)}
baseline delivery cost is therefore
\revrev{}{policy independent and}
always incurred.  We obtain worst-case bounds on cost ratios
by assuming it to be zero.
\revrev{}{Setting it to zero also allows us to entirely focus on the policy dependent costs.} 
\revfour{}{Finally, note that a third party service's accounting for storage and
remote bandwidth costs would, in practice, be based on particular time,
size, and bandwidth granularities.  The finer-grained the accounting, the
more closely our model would correspond to the real system.}

At the time a request is made for a file object not currently in the cache,
the system must,
in an online fashion, decide
\revrev{(i) whether the file should be cached or not and (ii) for how long the object should be cached.}{whether the
  \revfour{file}{object}
  should be cached or not.}
Naturally, the total delivery cost of different caching policies
will depend substantially on the choices
\revrev{they make here}{made}
and the request patterns of consideration.

To illustrate the impact of these choices,
consider the most basic TTL-based cache policy that
inserts
\revrev{an}{a file}
object into the cache whenever a request is made for the object (and the object is not currently in the cache)
and
\revrev{stores a
  file for $T$ time units after it is requested.}{retains the
  \revfour{file}{object}
  until $T$ time units elapse with no requests.}
This policy would
\revfour{endure}{incur}
a total cost of $R+T$
if a single request is made for the
\revfour{file.}{object.}
However,
\revrev{assuming that it somehow was}{if it was}
known that the
\revfour{file}{object}
would only receive a single request,
it would be optimal to not cache the
\revfour{file}{object}
at all.
In this case, it is easy to see that the minimal delivery cost is $R$.
For this particular example,
the cost ratio between the basic TTL-based policy and the (offline) optimal
is therefore $\frac{R+T}{R}$.  In general, we want these cost ratios to be as small as possible
both for (i) worst-case
\revtwo{arrival}{request}
patterns where an adversary selects the
\revtwo{arrival}{request}
pattern
and (ii) average case scenarios with more realistic
\revtwo{arrival}{request}
patterns.
Section~\ref{sec:worse-case} and Sections~\ref{sec:general}-\ref{sec:general-policies}
provide worst-case and average-case analysis, respectively,
for different TTL-based {\em cache on $M^{th}$ request}
insertion policies (Section~\ref{sec:policies}).

\section{Insertion policies}\label{sec:policies}

In this paper,
we compare the delivery costs of different
{\em cache on $M^{th}$ request} insertion policies
\revrev{in a TTL-based system that uses the same single-threshold rule
to decide when an object should be evicted from the cache.
With this eviction rule,
an object currently in the cache remains in the cache for as long as there are new
requests to the object within $T$ time units of the most recent prior request to that object.
Whenever there has not been a request to the object within $T$ time units the object
is removed from the cache (to reduce storage costs for that object).}{when using a TTL-based eviction policy
  in which an object remains in the cache after insertion until a time interval $T$ has elapsed without any
  requests for the object.  Note that with elastic resources, eviction is not needed for making room for
  new objects, but instead is needed for reduction of storage costs.  As we show, a simple TTL rule is very effective for this purpose.}
We next describe the insertion policies considered in this paper.

\begin{itemize}

\item {\bf Always on $1^{st}$ ($T$):}
  Always cache
  \revtwo{the object if the requested object is}{a requested object if}
  not in the cache already
  and keep it in the cache until $T$ time units have passed since the most recent request.

\item {\bf Always on $M^{th}$ ($M,T$):}
  The system maintains a counter for how many times
  \revrev{each object have been requested
  since the object was most recently removed from the cache or since the start of the system. 
  Cache the request if the cache is empty for $M$ consecutive requests
  and keep it in the cache until $T$ time units have passed since the most recent request.}{each uncached object has been requested.
    When the counter reaches $M$ the object is cached, and is kept in the cache until $T$ time units have passed since the most recent
    request, at which point the object is evicted and the counter is reset to 0.}
  For $M = 1$, this corresponds to {\em always on $1^{st}$}.

\item {\bf Single-window on $M^{th}$ ($M,T$):}
  \revfour{\revrev{The system maintains a counter for how many times
  each object have been requested
  within a time window $T$ of the most recent prior request for that object.
  The counter is initialized to one the first time that a request is made to the object
  or when a request is made to an object that has not been requested within the last $T$ time units.
  The counter is incremented by one whenever the object is referenced within $T$ time units.
  Finally, when the counter reaches $M$, the object is cached.}{A counter is maintained for how many times each uncached object has been
  requested within a time window $T$ of the most recent request for that object (including the most recent request itself).  The object
  is inserted into the cache whenever the counter reaches $M$.}}{The system maintains a counter for each uncached object
    that has been requested at least once within the last $T$ time units.
    The respective
    counter is initialized to one the first time that a request is made to an object
  or when a request is made to an object that has not been requested within the last $T$ time units.
  The counter is incremented by one whenever the object is referenced within $T$ time units of its most recent previous request.
  Finally, when the counter reaches $M$, the object is cached.}
Again,
\revtwo{the object remains in the cache as long as the object is requested
  again within $T$ time units,
  and, if no new request is made within $T$ time units,
  the object is removed from the cache.}{the object remains in the cache until a time
  interval $T$ has elapsed without any requests for the object.}
  For $M$=$1$, this policy corresponds to {\em always on $1^{st}$}.

\item {\bf Dual-window on $2^{nd}$ ($W,T$):}
  \revfour{Rather than requiring $M > 2$ requests being accumulated consecutively within $T$ time units of each other,
  a tighter time threshold $W \le T$ can be
  \revtwo{used as an indicator for when the request rate may
    be sufficient to motivate cache insertion.
  To capture such
  \revrev{more aggressive}{a}
  policy, we consider a two-parameter
  {\em dual-window on $2^{nd}$} version of the {\em single-window on $2^{nd}$} policy.
  With this policy,}{used.  With the {\em dual-window on $2^{nd}$} policy,}}{This policy is similar to {\em single-window on 2$^{nd}$},
      but uses a potentially tighter time threshold $W \le T$ for determining when
      to add an object to the cache. With the {\em dual-window on $2^{nd}$} policy,}
  \revrev{the system caches the object if the
  requested object has been requested at least once within the
  last $W$ time units and keeps the object}{when an uncached object is requested it is added to the cache if there has been a
    previous request for the object within the last $W$ time units, and is kept}
  in the cache until
  $T$ time units have passed since the most recent request.  
  \revtwo{Note that this}{This}
  policy reduces to the
  basic {\em single-window on $2^{nd}$}
  \revtwo{version when}{when}
  $W = T$. 
  
\end{itemize}

\section{Worst case bounds}\label{sec:worse-case}

For this analysis we consider an arbitrary request sequence $\mathcal{A} = \{a_i\}$ for a single object with $N$ requests,
where $a_i$ is the inter-request time between requests
\revrev{$i$ and $i-1$, or the time since the start of the trace (in the case it is the first request).
  We also assume that the system starts with an empty cache.}{$i$ and $i-1$
  \revfour{($2 \le i \le N$), with $a_1$ defined as 0.}{($2 \le i \le N$).}
  We assume that the object is initially uncached.}

\subsection{Offline optimal lower bound}

\revrev{}{We first derive the cost expression for the optimal (offline) caching policy
  (across all possible policy classes; not restricted to TTL-based policies)
  for the case
  when the cache has perfect prior knowledge of the request sequence $\mathcal{A}$.}
\revtwo{Assuming that
\revrev{we start with an empty cache,}{the object is initially uncached,}
the}{The}
first request will always
\revtwo{need to endure}{incur}
a remote bandwidth cost $R$.
For each of the later requests $i$ ($2 \le i \le N$),
in the (offline) optimal case,
the object should have been cached (if not already in the cache) at the time of the $(i-1)^{st}$ request and
\revtwo{should remain in the cache at least until}{remain retained until at least}
the $i^{th}$ request, whenever $a_i < R$.
On the other hand,
\revtwo{if the inter-request time is greater that $R$ (i.e., $a_i > R$),}{if $a_i > R$,}
the object should not have been cached at the time of the $i-1^{st}$ request,
or should have been dropped from the cache (if it was already in the cache) just after serving request $i-1$.
In this case, the $i^{th}$ request should
\revtwo{instead endure}{incur}
the remote bandwidth
\revtwo{cost $R$ at the time of request $i$.}{cost $R$.}
\revtwo{Given these observations, we can formulate the following lemma regarding the (offline) optimal cost.}{The
  following lemma regarding the (offline) optimal cost follows directly from these observations.}
  
\begin{lemma}\label{lem:opt}
  \revrev{Given the above assumptions,}{Given an arbitrary request sequence $\mathcal{A}$,}
  the minimum total delivery cost of
  \revrev{any}{the}
  \revrev{(offline) optimal}{optimal offline}
  policy is:
  \begin{align}\label{eqn:general_opt}
    C_{opt}^{offline} = R + \sum_{i=2}^N \min[a_i,R].
  \end{align}
  \end{lemma}
\vspace{-6pt}
\revtwo{\begin{proof}
  \revrev{This proof}{The result}
  follows directly from the above observations
  regarding the
\revfour{optimal offline}{{\em optimal offline}}
  policy.
\end{proof}}{}

\revrev{}{Lemma~\ref{lem:opt} provides a fundamental {\em offline bound} for all caching policies.}
We next
\revrev{use equation (\ref{eqn:general_opt}) of Lemma~\ref{lem:opt} to derive}{derive}
worst-case bounds for the various online policies outlined
in Section~\ref{sec:policies}.

\subsection{Always on $1^{st}$ ($T$)}

For an arbitrary request sequence $\mathcal{A}$,
this (online) policy
\revfour{endures}{incurs}
a total delivery cost equal to:
\vspace{-4pt}
\begin{align} \label{eqn:general_k1}
  C^{always}_{M=1,T} = R + T + \sum_{i=2}^N x_i,
\end{align}
\vspace{-8pt}
where
\begin{align} \label{eqn:general_k1_xi}
  x_i = \left\{ \begin{array}{ll}
    T+R, & \textrm{if}~a_i > T\\
    a_i, & \textrm{otherwise}.\\
  \end{array}\right.
\end{align}
Here, and throughout the paper,
we use the superscript
\revrev{}{on the cost $C$}
to indicate the class of insertion policy,
the subscript to indicate the parameters being used by the policy,
and potential parameter assignment to indicate potential special cases considered.
\revrev{}{In equation (\ref{eqn:general_k1}), the $R$ term corresponds to the cost of retrieving
  a copy of the
  \revfour{file}{object}
  to serve the first request in the sequence and the $T$ term corresponds to the cache storage cost
  \revfour{endured}{incurred}
  after the last request.
  For requests $2 \le i \le N$, equation (\ref{eqn:general_k1_xi})
  then takes into account whether request $i$ occurs within $T$ of the prior request
  \revtwo{(at which point an extra storage cost of $a_i$ is added)}{(implying an additional storage cost of $a_i$)}
  or the object has been removed from the cache prior to the
  \revtwo{request.  For this second case,
  the cache endured a storage cost $T$ before the object was released from the
  cache and a bandwidth cost $R$ is needed to retrieve a new copy.}{request
    (implying an additional storage cost $T$ before the object was evicted
    and a bandwidth cost $R$ to retrieve a new copy).}}
Given equations
\revrev{(\ref{eqn:general_opt}) and (\ref{eqn:general_k1}),}{(\ref{eqn:general_opt})-(\ref{eqn:general_k1_xi}),}
it is now possible to show the following
theorem.

\begin{theorem}\label{thm:always-1st}
  The best (optimal) competitive ratio using
  {\em always on $1^{st}$} is achieved with $T=R$ and is equal to 2.
  More specifically,
  \begin{align}
    \max_{\mathcal{A}} \frac{C^{always}_{M=1,T=R}}{C_{opt}^{offline}} \le \max_{\mathcal{A}} \frac{C^{always}_{M=1,T}}{C_{opt}^{offline}}
  \end{align}
  for all $T$, and $\frac{C^{always}_{M=1,T=R}}{C_{opt}^{offline}} \le 2$ for
  all possible
  sequences $\mathcal{A} = \{a_i\}$.
\end{theorem}

\begin{proof}
  \revtwo{For this proof, we}{We}
  consider an arbitrary
  \revrev{arrival pattern}{request sequence}
  \revfour{$\mathcal{A}$}{$\mathcal{A}$ with $N$ requests}
  and
  then bound the cost ratio based on the worst-case patterns that an adversary could create.
  For this and the
  \revrev{proceeding}{following}
  proofs we note that the first request always must
  \revfour{endure a (minimum)}{incur a remote bandwidth}
  cost $R$ and then focus on the
  worst-case pattern of the remaining
  \revfour{$|\mathcal{A}| - 1$}{$N-1$}
  requests.
  
  Case $T \le R$:
  For the remaining
  \revfour{$|\mathcal{A}| - 1$}{$N-1$}
  requests,
  let us define the following sets:
  $S = \{i | a_i \le T \}$,
  $S' = \{i | T < a_i \le R \}$, and
  $S'' = \{i | R < a_i \}$.
  Note that the set $S$ consists of those requests that
  would result in cache hits, if using {\em always on $1^{st}$},
  while the requests in the other sets would result in cache misses.
  Also, note that the requests in both set $S$ and $S'$ would result in
  the {\em optimal offline} policy retrieving the object from the local cache.
  Now, for any
  \revrev{arrival pattern}{request sequence}
  $\mathcal{A}$,
  we have the following
  \revfour{conditions:}{relations:}
  \begin{align}\label{eqn:always-one-proof-step1}
    \frac{C^{always}_{M=1,T}}{C_{opt}^{offline}}
    & = \frac{R + \sum_{i \in S} a_i + (|S'|+|S''|) (T+R) + T}{ R + \sum_{i \in S} a_i + \sum_{i \in S'} a_i + |S''| R} \nonumber\\
    & \le \frac{(R + T)(1+|S''|) + (R + T)|S'|}{ R(1+|S''|) + \sum_{i \in S'} a_i} \nonumber\\
    & \le \frac{(R + T)(1+|S''|) + (R + T)|S'|}{ R(1+|S''|) + |S'| T }
    \le \frac{R + T}{T}.
  \end{align}
  \revrev{In each of the three steps}{To establish the three inequalities in (\ref{eqn:always-one-proof-step1})}
  we have used that:
  (i) $\frac{X+\sum_{i \in S} a_i}{X(1-\epsilon)+\sum_{i \in S} a_i} \le \frac{X}{X(1-\epsilon)}$
  for $0 \le \epsilon \le 1$ and $\sum_{i \in S} a_i \ge 0$,
  (ii) $T \le a_i$ when $i \in S'$, and
  (iii) $\frac{d}{dx} (\frac{R+T}{R(1-x) + xT}) = - \frac{(R+T)(T-R)}{(R+x(T-R))^2} \ge 0$
  \revrev{when $T \le R$.}{when $T \le R$, respectively.}
  Clearly, since $\frac{R + T}{T}$ is monotonically decreasing for the range $0 \le T \le R$,
  the (above) worst-case bound is tightest when $T \rightarrow R$ (equal to 2).

  Case $R \le T$:
  Let us define the following sets for
  \revfour{$2 \le i \le |\mathcal{A}|$:}{$2 \le i \le N$:}
  $G = \{i | a_i < R \}$,
  $G' = \{i | R \le a_i \le T \}$, and
  $G'' = \{i | T < a_i \}$.
  Here, sets $G$ and $G'$ consist of those requests that
  would result in cache hits with {\em always on $1^{st}$},
  but only the requests in set $G$ would result in cache hits with the {\em optimal offline} policy. 
  Using
  \revfour{similar inequalities as above,}{a similar approach as for the first case,}
  we obtain the following:
  \begin{align}
    \frac{C^{always}_{M=1,T}}{C_{opt}^{offline}}
    & = \frac{R + \sum_{i \in G} a_i + \sum_{i \in G'} a_i + |G''| (T+R) + T}{ R + \sum_{i \in G} a_i + (|G'|+|G''|) R} \nonumber\\
    & \le \frac{(R + T)(1+|G''|) + \sum_{i \in G'} a_i}{ R(1+|G''|) + |G'| R} \nonumber\\
    & \le \frac{(R + T)(1+|G''|) + T|G'|}{ R(1+|G''|) + |G'| R }
    \le \frac{R + T}{R}.
  \end{align}
  Here,
  \revfour{\revrev{step}{inequality}
    (i) is}{the first inequality is}
  derived in the same way as the first inequality in
  \revfour{(\ref{eqn:always-one-proof-step1})
  \revrev{step}{inequality}
  (ii) uses}{(\ref{eqn:always-one-proof-step1}), the second inequality uses the fact}
  that $a_i \le T$ when $i \in G'$,
  and
  \revfour{step (iii) uses that}{the third inequality uses the fact that}
  $\frac{d}{dx} (\frac{(R+T)(1-x)+Tx}{R}) = -1 < 0$.
  Now, since $\frac{R + T}{R}$ has its minimum in the range $R \le T$ when $T=R$,
  we have that $T=R$ provides the tightest bound (equal to 2).

  Finally, inserting $T=R$ into either of the two bounds,
  we obtain the worst-case bound of 2.
  The bound is tight and is achieved,
  \revfour{e.g.,}{for example,}
  when requests are
  evenly spaced by $T+\epsilon$, for some $\epsilon > 0$.
  In this case, $|S| = |G| = |G'| = 0$ and $\frac{C^{awlays}_{M=1,T=R}}{C_{opt}^{offline}} = \frac{T+R}{R} = 2$.
\end{proof}

\subsection{Always on $M^{th}$ ($M,T$)}

By generalizing the techniques used to prove the worst-case properties
of {\em always on $1^{st}$} to consider additional counter
\revrev{states that the cache can be in before the file can enter the cache,}{states,}
it is possible to prove the following theorem.

\begin{theorem}\label{thm:always-mth}
  The best (optimal) competitive ratio using
  the {\em always on $M^{th}$} policy
  is achieved with $T=R$ and is equal to $M+1$.
  More specifically,
  \begin{align}
    \max_{\mathcal{A}} \frac{C^{always}_{M,T=R}}{C_{opt}^{offline}} \le \max_{\mathcal{A}} \frac{C^{always}_{M,T}}{C_{opt}^{offline}}
  \end{align}
  for all $T$, and $\frac{C^{always}_{M,T=R}}{C_{opt}^{offline}} \le M+1$ for
  all possible
  sequences $\mathcal{A} = \{a_i\}$.
\end{theorem}

A proof for Theorem~\ref{thm:always-mth} is provided in the Appendix.
Similar to the proof for {\em always on $1^{st}$},
the proof identifies sets of inter-request times $a_i$ based on
differences and similarities in how the
{\em always on $M^{th}$} policy and the {\em optimal offline} policy treat these sets of requests.
In particular, sets are defined based on the states of the {\em always on $M^{th}$} policy
\revrev{(depending on the number of requests to the object since the object was removed from the cache
  most recently or the request sequence started, with states associated with counts of $M$ or higher,
  all aggregated into an ``in-cache'' state)}{(depending on the object's caching status and, if uncached, counter value)}
and how $a_i$ relates to $T$ and $R$.
This generalizes the number of (mutually exclusive) sets of requests
from $2 \times 3$ for the {\em always on $1^{st}$} policy,
to $2 \times (2M+1)$ for the general {\em always on $M^{th}$} policy,
where $2M+1$ sets are needed for each of the two cases when $T \le R$ and $R \le T$, respectively.

Using this proof method, we also identify a request pattern that
shows that the bound is tight.  In particular, the worst-case bound is
achievable by a request pattern in which requests occurs in batches of
$M$ requests,\footnote{Here, we consider a ``batch'' to consist of sufficiently closely spaced requests
  that
  \revrev{there is negligible inter-request times between them,}{the inter-request times are negligible,}
  but where 
  the requests still are treated as individual requests, and the cache still needs to make individual decisions whether to cache
  or not to cache the object at the time of each of these requests.}
\revrev{where each such batch is spaced}{and the batches are separated}
by more than
\revrev{$R$ time units (when $T \le R$)
or more than $T$ (when  $R \le T$); whichever is greater.}{$\max[R, T]$ time units.}
To see this, let us consider the $T \le R$ case.
In this case,
with the above
\revrev{arrival pattern,}{request sequence,}
in each batch cycle,
the {\em always on $M^{th}$} policy downloads the object $M$ times,
finally stores a copy at the time of the $M^{th}$ request,
and then keeps it in the cache for $R$ time units.
This pattern results in a total cost of $(M+1)R$ per batch.
In contrast, the {\em optimal offline} policy downloads a single copy (at cost $R$),
serves all $M$ requests using this copy, and then
\revrev{instantaneously deletes the copy (to avoid storage costs).}{immediately deletes the copy, incurring negligible storage costs.}
The
\revfour{arguments}{argument}
for the $R \le T$ case is analogous.

\subsection{Single-window on $M^{th}$ ($M,T$)}

While the number of counter states to consider is the same
for {\em single-window on $M^{th}$} as for {\em always on $M^{th}$},
the possible state transitions when the counter is below $M$
\revrev{differ, as the counter is reset each time there is no request within a window $T$.}{differ
  (e.g., counter is reset each time there is no request within a window $T$).}
To account for this,
\revtwo{the worst-case proof of the {\em single-window on $M^{th}$} policy}{our proof of the following
  theorem for the {\em single-window on $M^{th}$} policy}
requires $2 \times (M-1)$ additional sets to be defined
($M-1$ for when $T \le R$ and $M-1$ for when $R \le T$).
\revtwo{While this somewhat complicates the proof,
it is possible to prove the following theorem.}{}

\begin{theorem}\label{thm:cache-2nd-two}
  The best (optimal) competitive ratio using
  the {\em single-window on $M^{th}$} policy
  is achieved with $T=R$ and is equal to $M+1$.
  More specifically,
  \begin{align}
    \max_{\mathcal{A}} \frac{C^{window}_{M,T=R}}{C_{opt}^{offline}} \le \max_{\mathcal{A}} \frac{C^{window}_{M,T}}{C_{opt}^{offline}}
  \end{align}
  for all $T$, and $\frac{C^{window}_{M,T=R}}{C_{opt}^{offline}} \le M+1$ for
  all possible
  sequences $\mathcal{A} = \{a_i\}$.
\end{theorem}

A proof for Theorem~\ref{thm:cache-2nd-two} is provided in the Appendix.
Interestingly, the same request pattern, with batches of size $M$ separated by at least $\max[T,R]$,
as used to show that Theorem~\ref{thm:always-mth} is tight, provides proof that Theorem~\ref{thm:cache-2nd-two} is tight.

\subsection{Dual-window on $2^{nd}$ ($W,T$)}

Using similar methods as used in prior subsections (this time based on $3 \times 8$ sets,
accounting for the relationship of $a_i$ to $W$, $T$ and $R$),
\revtwo{we show}{it is possible to prove the following theorem establishing}
\revtwo{that the worst-case bound of the
two-parameter dual-window version of the {\em cache on $2^{nd}$ request} policy
is no better than the one-parameter single-window version.
More specifically,}{that}
{\em dual-window on $2^{nd}$}
has the same worst-case properties as {\em single-window on $2^{nd}$}. 
\revtwo{This result is summarized in the following theorem and proven}{A proof is provided}
in the Appendix.

\begin{theorem}\label{thm:cache-Mth}
  The best (optimal) competitive ratio using
  the {\em dual-window on $2^{nd}$} policy
  is achieved with $T=W=R$ and is equal to 3.
  More specifically,
  \begin{align}
    \max_{\mathcal{A}} \frac{C^{window}_{M=2,W=R,T=R}}{C_{opt}^{offline}} \le \max_{\mathcal{A}} \frac{C^{window}_{M=2,W,T}}{C_{opt}^{offline}}
  \end{align}
  for all $W$ and $T$, and $\frac{C^{window}_{M=2,W=R,T=R}}{C_{opt}^{offline}} \le 3$ for
  all possible
  \revtwo{arrival}{request}
  sequences $\mathcal{A} = \{a_i\}$.
\end{theorem}

\section{Steady-state: Offline Bound}\label{sec:general}

\revrev{}{Thus far our results have not made any restrictions to the request sequences.}
For the remaining analysis in this paper,
we assume that inter-request times are independent and identically distributed.
Under this assumption,
we derive expressions for a general inter-request time distribution $f(t)$
with cumulative distribution function $F(t)$,
as well as for specific example distributions.
In the following,
we let $E[a_i]$ denote the average inter-request time, we let
\begin{align}\label{eqn:EltX}
  E[a_i | a_i \le X ] & = \frac{\int_0^{X}tf(t)\textrm{dt}}{\int_0^{X}f(t)\textrm{dt}}
  = X - \frac{1}{F(X)}\int_0^{X}F(t)\textrm{dt}
\end{align}
denote the average inter-request time
given that the inter-request time
is no more than $X$ time units,  and we let
\begin{align}\label{eqn:P_ltX_gtY}
  P(a \le X | a > Y) & = \frac{\int_Y^{X}f(t)\textrm{dt}}{\int_Y^{\infty}f(t)\textrm{dt}}
  = \frac{F(X)-F(Y)}{1-F(Y)}
\end{align}
denote the (conditional) probability that an inter-request time is
\revfour{less than}{no more than}
$X$ given that the inter-request time is
\revfour{at least}{greater than}
$Y$ time units.
\revrev{}{In this section we derive results for the
\revfour{optimal offline}{{\em optimal offline}}
  policy,
  while in Sections 6-8 we consider online policies.}

\subsection{\revrev{Offline lower bound}{General inter-request time distribution}}

Throughout this analysis we will derive expressions
for the average cost per time unit.
For the (optimal) offline policy,
this cost can be calculated as
the expected cost associated with an arbitrary request
divided by the average
\revtwo{inter-arrival}{inter-request}
time $E[a_i]$:
\begin{align}\label{eqn:offline-opt-general}
  C_{opt}^{offline} & = \frac{1}{E[a_i]} \left[ \int_{0}^{R} t f(t) \textrm{dt} + R \int_{R}^{\infty} f(t)\textrm{dt}  \right] \nonumber\\
  & = \frac{1}{E[a_i]} \left[ [t F(t)]_{0}^R - \int_{0}^{R} F(t) \textrm{dt} + R(1-F(R)) \right] \nonumber\\
  & = \frac{1}{E[a_i]} \left[ R - \int_{0}^{R} F(t) \textrm{dt} \right].
\end{align}
Here, we associate all requests with inter-request times $t$ less than $R$
with the cost $t$ to keep the object in the cache for an additional $t$ time units (first integral in the first line),
while all other requests (with $R < t$)
\revfour{endure}{incur}
a cost $R$ (second integral in the first line).
We then use integration by parts (step 2) and algebraic simplifications
(step 3) to derive the final expression.

\subsection{Example distributions}

We next consider four example distributions.

{\bf Exponential:}
Assuming a Poisson process, with exponential inter-request times, we have
\begin{align}
  \label{eqn:exponential-A}
  f(t) & = \lambda e^{-\lambda t}, \quad 
  F(t) = 1 -  e^{-\lambda t}, \quad 
  E[a_i] = \frac{1}{\lambda}\\
  \label{eqn:exponential-B}
  \int_{0}^{t} F(t) \textrm{dt} & = t - \frac{1-e^{-\lambda t}}{\lambda}.
\end{align}
\revfour{Though}{Through}
insertion of these equations into
equation (\ref{eqn:offline-opt-general})
we obtain the following cost function:
\begin{align}
  \label{eqn:offline-opt-general-EXP}
  C_{opt}^{offline}
  & = \lambda \left[ R - \frac{1-e^{-\lambda R}}{\lambda} \right]
  = 1-e^{-\lambda R}.
  \end{align}  

{\bf Erlang:}
  We next consider Erlang distributed inter-request times
  with shape parameter $k$ (integer) and rate parameter  $\lambda > 0$:
  \begin{align}
    \label{eqn:erlang-A}
    & f(t) =  \frac{\lambda^k t^{k-1}e^{-\lambda t}}{(k-1)!}, ~
    F(t) = 1 - \sum_{n=0}^{k-1}\frac{1}{n!}e^{-\lambda t}(\lambda t)^n, ~
    E[a_i] = \frac{k}{\lambda},\\
    \label{eqn:erlang-B}
    & \int_{0}^{t} F(t) \textrm{dt}
    = t - \frac{k}{\lambda} + \frac{e^{-\lambda t}}{\lambda} \sum_{m=1}^{k} \sum_{n=0}^{m-1} \frac{(\lambda t)^n}{n!}.
  \end{align}
  \revrev{Though insertion of these equations into
  equation (\ref{eqn:offline-opt-general})
  we obtain the following cost function:}{Substitution into equation (\ref{eqn:offline-opt-general}) yields:}
  \begin{align}
    \label{eqn:offline-opt-general-ERLANG}
    C_{opt}^{offline}
    & = 1 - \frac{e^{-\lambda R}}{k} \sum_{m=1}^{k} \sum_{n=0}^{m-1} \frac{(\lambda R)^n}{n!}.
  \end{align}
  
{\bf Deterministic:}
\revrev{Taking this to the extreme, we next consider the case}{In the extreme case for low variability,}
all inter-request times are equal to a constant $a$.
  Let $\delta_a(t)$ and $u_a(t)$ represent the Dirac delta function and the unit step function,
  both with (unit) singularities at $t=a$.  Then, we have:
  \begin{align}
    \label{eqn:deterministic-A}
    f(t) & =  \delta_a(t), \quad 
    F(t) = u_a(t), \quad 
    E[a_i] = a,\\
    \label{eqn:deterministic-B}
    \int_{0}^{t} F(t) \textrm{dt} & = \max[0,t-a].
  \end{align}
  \revrev{Though insertion of these equations into
  equation (\ref{eqn:offline-opt-general})
  we obtain the following cost function:}{Substitution into equation (\ref{eqn:offline-opt-general}) yields:}
  \begin{align}
    \label{eqn:offline-opt-general-DETERMINISTIC}
    C_{opt}^{offline}
    & = \min[1,\frac{R}{a}].
  \end{align}

  {\bf Pareto:}
  Finally, we consider Pareto distributed inter-request times (as an example of heavy-tailed distributions)
  with shape parameter $\alpha > 1$ (when $0 < \alpha \le 1$ the expected inter-request time is infinite)
  and scale parameter $t_m > 0$.  In this case, we have:
  \revfour{\begin{align}
    \label{eqn:pareto-A}
    f(t) & =  \frac{\alpha t_m^{\alpha}}{t^{\alpha+1}}, \quad 
    F(t) = 1 - \left(\frac{t_m}{t}\right)^{\alpha}, \quad 
    E[a_i] = \frac{\alpha t_m}{\alpha-1}\\
    \label{eqn:pareto-B}
    \int_{0}^{t} F(t) \textrm{dt}
    & = \left\{ \begin{array}{ll}
      t + \frac{t\left(\frac{t_m}{t}\right)^{\alpha} - t_m\alpha}{\alpha-1}, & t_m \le t\\
      0, & t < t_m.\\
      \end{array}\right.
  \end{align}}{\begin{align}
      \label{eqn:pareto-A1}
      f(t) & =  \frac{\alpha t_m^{\alpha}}{t^{\alpha+1}}, \quad 
      F(t) = 1 - \left(\frac{t_m}{t}\right)^{\alpha}, \quad 
      t_m \le t,\\
      \label{eqn:pareto-A2}
      \quad\quad E[a_i] = \frac{\alpha t_m}{\alpha-1},\\
      \label{eqn:pareto-B}
      \int_{0}^{t} F(t) \textrm{dt}
      & = \left\{ \begin{array}{ll}
        t + \frac{t\left(\frac{t_m}{t}\right)^{\alpha} - t_m\alpha}{\alpha-1}, & t_m \le t\\
        0, & t < t_m.\\
      \end{array}\right.
    \end{align}}
  \revrev{Though insertion of these equations into
  equation (\ref{eqn:offline-opt-general})
  we obtain the following cost
  \revrev{function:}{function:}}{Substitution into equation (\ref{eqn:offline-opt-general}) yields:}
  \begin{align}
    \label{eqn:offline-opt-general-PARETO}
    C_{opt}^{offline}
    & = \left\{ \begin{array}{ll}
      1 - \frac{1}{\alpha} \left( \frac{t_m}{R} \right)^{\alpha - 1}, & t_m \le R\\
      \frac{R(\alpha-1)}{\alpha t_m}, & R < t_m.\\
      \end{array}\right.
  \end{align}
  \revtwo{\revrev{}{For the first of these expressions we have used that}
    \begin{align}
    \frac{\alpha-1}{\alpha t_m} \left(R - \frac{R\left(\left(\frac{t_m}{R}\right)^{\alpha} + \alpha - 1\right) - t_m \alpha}{\alpha-1} \right) = 1 - \frac{1}{\alpha} \left( \frac{t_m}{R} \right)^{\alpha - 1}.
\end{align}}{}
    
\section{Steady-state: Static Baseline Policy with Known \revfour{Request Rate}{Inter-request Distribution}}\label{sec:baselin}

To provide some estimates for the best possible {\em online}
\revrev{}{cache}
\revrev{performance to compare different TTL-based insertion
\revrev{policies,}{policies against,}}{performance,}
in this
\revfour{section, we}{section we}
consider the
case when
an ``oracle'' provider knows
\revfour{both the request distribution and
the exact parameters specifying the request intensity
for that particular object.}{the precise inter-request time distribution for each object.}
For this case, we consider
\revrev{the delivery cost of an ``oracle'' provider that tries to minimize}{a {\em static baseline} policy that tries to minimize}
the delivery cost by selecting between
the extremes of (i) always keeping
the object in the cache, or
(ii) never caching the object.

\subsection{Optimal online policy when non-decreasing hazard rate}

\revfour{Interestingly,
\revrev{it appears that the}{the}
optimized versions of the different
{\em cache on $M^{th}$ request} policies considered in this paper reduce to this
{\em static baseline}
policy when
\revrev{}{(i) inter-request distributions are short-tailed (deterministic,
  \revfour{Erland,}{Erlang,}
  and exponential), (ii)}
request rates are
\revrev{known and}{known, and (iii)}
we are allowed to optimize over the
policy parameters
\revrev{(i.e., selecting optimized $W$ and $T$) for the case
when the inter-requests are exponential and deterministic.}{$W$ and $T$.}
\revrev{A discussion of this is presented in Section~\ref{sec:results}.
Based on these findings we conjecture that {\em static baseline}
provides an online bound when the distribution parameters are known
and inter-request distributions has
a non-decreasing hazard rate
(in our case exponential, Erlang, and deterministic).
This conjecture has not yet been proven and remains interesting future work.}{Even more interesting,
  {\em static baseline}
  provides an {\em online bound} when the inter-request distribution parameters are known
  and the distribution has
an increasing or constant
  hazard rate.
  This is formalized and proven in the following theorem.}}{Interestingly,
  the {\em static baseline}
  provides an {\em online bound} when the inter-request distribution parameters are known
  and the distribution has an increasing or constant
  hazard rate.}

\begin{theorem}\label{thm:non-decreasing-hazard-rate}
  {\em Static baseline}
  \revrev{is the {\em optimal online} policy}{achieves the minimum cost of any {\em online policy}}
  when the inter-request distribution has
  \revrev{a non-decreasing}{an increasing or constant}
  hazard rate. 
\end{theorem}

\begin{proof}
  \revfour{With IID inter-request times, need only consider a single inter-request time;
    between requests $i-1$ and $i$, say.}{Since inter-request times are IID,
    we need consider only a single representative inter-request time between requests $i-1$ and $i$, for some $i \geq 2$.}
  After servicing request $i-1$,
  any online policy will, at each subsequent instant of time up to
  \revfour{}{the}
  time
  of request $i$ or until the object is discarded, need to decide whether to
  retain the object in the cache,
  or evict it.  The only information the online policy can use to make
  this decision is the
\revfour{}{elapsed}
time since request $i-1$.
\revfour{So,}{Therefore,}
any online policy
  will have a
  \revrev{}{threshold}
  parameter $t^*$, such that as long as the time since request $i-1$
  \revfour{is $< t^*$,}{is less than $t^*$,}
  the object
  \revfour{will keep being retained.}{is retained.}
  If time $t^*$
  \revfour{is reached without yet}{elapses before}
  getting request $i$,
  the object is evicted.  Letting $C(t^*)$
  denote the expected cost incurred from after servicing request $i-1$,
  up to and including the servicing of request $i$, we have:
  \begin{align}
  C(t^*) & = F(t^*) E[a_i | a_i \le t^*] + (1 - F(t^*)) (R + t^*).
  \end{align}
  Using expression (\ref{eqn:EltX}) and simplifying gives:
  \begin{align}\label{eqn:Ct-star}
  C(t^*) & = R (1 - F(t^*)) + t^* - \int_0^{t^*} F(t)\textrm{dt}.
  \end{align}
  Taking the derivative with respect to $t^*$ gives:
  \begin{align}
  \frac{dC(t^*)}{dt^*} & = 1 - F(t^*) - R f(t^*).
  \end{align}
  \revrev{The}{A constant hazard rate corresponds to an exponential distribution,
    and for this case it is straightforward to show that the derivative is negative
    for all $t^*$, positive for all $t^*$, or is constant at 0 for all $t^*$ (when $R \lambda = 1$),
    implying that the {\em static baseline} policy achieves minimum cost.  Consider now the case
    of increasing hazard rate, and note that the}
  derivative is zero
  \revrev{at the point}{when}
  $R = \frac{1 - F(t^*)}{f(t^*)}$.

  \revrev{Now, need to test whether this is a local minimum or maximum.
    Taking the second derivative yields:}{Whether such a point is a minimum or maximum depends
    on the second
    \revtwo{derivative at that
    \revtwo{point.  The second derivative is}{point;} given by}{derivative, given by}}
  \begin{align}
  \frac{d^2C(t^*)}{d^2t^*} = -f(t^*) - R \frac{df(t^*)}{dt^*}.
  \end{align}
  \revfour{This is less than zero at
  \revrev{the point}{a point where}}{At a point where}
  $R = \frac{1 - F(t^*)}{f(t^*)}$,
  \revfour{exactly}{the second derivative is less than zero exactly}
  when the derivative of the hazard rate at
  \revrev{$t^*$}{this point}
  (the derivative of $f(t^*)/(1 - F(t^*))$) is positive.
  And so, when there is an increasing hazard rate,
  \revrev{the point}{any point \revtwo{where $t^*$ is such that}{where}}
  $R = \frac{1 - F(t^*)}{f(t^*)}$ is a local cost maximum, and
  the minimum cost must occur for $t^*=0$ or $t^* \rightarrow \infty$.
\end{proof}

\begin{corollary}\label{cor:corollary}
\revrev{For the case of non-decreasing hazard rate,
the {\em minimum} cost for any
{\em online policy},
would be to set $t^*$ such that $R = \frac{1 - F(t^*)}{f(t^*)}$.}{For inter-request time distributions
  such that (i) there is a unique value of $t^*$ where $R = (1-F(t^*))/f(t^*)$, and
  (ii) the derivative of the hazard rate at this value is negative,
  the minimum cost over all online policies is achieved with $t^*$ set to this value.}
\end{corollary}

\revfour{}{Note that the {\em cache on $M^{th}$} policies
  are identical to the {\em static baseline} if $T$ (and $W$ in the case of {\em dual-window})
  are chosen to be either 0 or $\infty$, whichever gives the best performance.
  Therefore, since {\em static baseline} provides an {\em online bound} when the inter-request distribution
  parameters are known and the distribution has an increasing or constant hazard rate,
  also the {\em cache on $M^{th}$} policies with optimized parameters achieve this bound in this case.}

\revrev{With}{In contrast to
  \revfour{with}{the case of}
  the short-tailed distributions (deterministic, Erlang, and exponential),
  for which {\em static baseline} is the {\em optimal online} policy, with}
Pareto (and other heavy-tailed distributions) the
\revrev{worst-case performance of}{competitive ratio of}
\revrev{this policy is unbounded.}{{\em static baseline} is unbounded
  \revfour{(e.g., Theorem~\ref{thm:static-pareto})}{(see Theorem~\ref{thm:static-pareto} for the case of the Pareto distribution)}
  even when request rates are known.
  For a Pareto distribution,
  \revfour{using (\ref{eqn:pareto-A})}{using (\ref{eqn:pareto-A1})}
  to substitute for $F(t^*)$ and $f(t^*)$ in
  $R = (1 - F(t^*))/f(t^*)$ yields $t^* = R \alpha$, under the condition that
  $t^* = R \alpha \ge t_m$.  Since Pareto has decreasing hazard rate for
  $t \ge t_m$, applying Corollary~\ref{cor:corollary} the
  \revfour{optimal online}{{\em optimal online}}
  policy
  for a Pareto inter-request time distribution sets $t^* = R \alpha$ when $R \alpha \ge t_m$.
  \revfour{And so, for $\alpha \rightarrow 1$, {\em always on $1^{st}$} with $T = R$ is the
    \revfour{optimal online}{{\em optimal online}}
    policy (when $t_m \le R$).}{And so, {\em always on $1^{st}$} with $T = R \alpha$ is the
    \revfour{optimal online}{{\em optimal online}}
    policy when $t_m \le R \alpha$.}
  Also,
  applying (\ref{eqn:Ct-star}) with the optimal $t^*$, for
  general $\alpha$ (and $t_m \le t^*$),
   \revthree{it appears that the}{it can be shown that the}
   competitive ratio of the
\revfour{optimal online}{{\em optimal online}}
   policy
  is at most 2 (attained when $\alpha \rightarrow 1$).}

Of course, in practice, the request rates of individual objects are never known exactly.
Therefore, the
\revtwo{baseline performance presented in this section can best be seen as}{{\em static baseline} policy is best seen as providing}
bounds
\revrev{of the possible online policies}{on the performance possible with an online policy}
\revfour{(when the distribution is exponential, Erlang, and deterministic)}{(when the inter-request distribution has an increasing or constant hazard rate)}
or as a general measurement stick.
Naturally, if the ``wrong'' choice is selected of these two
\revrev{extremes,}{extremes (always keep in cache or never cache),}
the worst-case performance ratio (regardless of distribution!) is unbounded.
\revfour{In later sections}{In Sections 7 and 8}
\revrev{we use these baseline online policies to}{we}
evaluate
different {\em online} insertion policies,
and their robustness over the full parameter
\revfour{region (regardless of how bad/good the system
can estimate the
\revtwo{arrival}{request}
rates of individual objects).}{space when the object inter-request distribution is unknown.}

\subsection{Exponential with known $\lambda$}

\revrev{For the special case when the Poisson arrival rates are known,
it is optimal to operate in one of two extremes.
In particular,
for arrival rates less than $\frac{1}{R}$ the
\revfour{file}{file object}
should never be cached, resulting in
an average cost per time unit of $\lambda R$,
and for higher request rates
than $\frac{1}{R}$ the
\revfour{file}{object}
should never be removed from the cache, resulting in a delivery cost of $1$.
The optimal {\em static baseline} policy under known Poisson arrivals is therefore:}{For the special case of a
  Poisson request process with known rate $\lambda$, the delivery cost with a static policy is minimized
  by never caching the
  \revfour{file}{object}
  if $\lambda < 1/R$,
  and always keeping the
  \revfour{file}{object}
  cached if $1/R \le \lambda$.
  The average cost per time unit in these two cases is given by $\lambda R$ and 1, respectively.
  The average cost per time unit of the {\em static baseline} policy for a Poisson
  request process with known rate is therefore:}
\begin{align}\label{eqn:online-opt-poisson}
  C_{opt}^{static} & = \min[\lambda R, 1].
\end{align}

This policy has the same cost as the
\revfour{optimal offline}{{\em optimal offline}}
policy in both asymptotes;
i.e., they both approach $\lambda R$ when $\lambda \rightarrow 0$ and
approach 1 when $\lambda \rightarrow \infty$.  However, given the ``wrong'' choice
of which of the two extremes should be used,
this otherwise ``optimal'' policy has an unbounded worst-case cost.
For example, consider the case that we have selected to never cache the object.
In this case, it is easy to see that the cost ratio compared to both the
{\em optimal offline} policy (equation (\ref{eqn:offline-opt-general-EXP})) and
the optimal {\em static baseline} policy (equation (\ref{eqn:online-opt-poisson})) is
\revfour{ubounded.}{unbounded.}
In particular, note that both $\frac{\lambda R}{1-e^{-\lambda R}}$ (comparing with {\em optimal offline})
and $\frac{\lambda R}{\min[\lambda R, 1]}$
(comparing with optimal {\em static baseline}) go to infinity as $\lambda \rightarrow \infty$.
Similarly, it is easy to see that for the case that we always cache a copy,
the ratio can be unbounded when request rates are low.
To see this, note that both $\frac{1}{1-e^{-\lambda R}}$ (comparing with {\em optimal offline})
and $\frac{1}{\min[\lambda R, 1]}$ (comparing with optimal {\em static baseline})
go to infinity as $\lambda \rightarrow 0$.

\revfour{Before analyzing more practical online policies that do not require perfect prediction of the
\revtwo{arrival}{request} rate,
we note
\revrev{that}{that in the case of a Poisson request process}
the worst-case
\revrev{bound}{competitive ratio}}{Assuming known inter-request time distribution,
for Poisson requests, the worst-case competitive ratio}
of the optimal {\em static baseline} policy is
\revrev{$\frac{e}{e-1}$,}{$\frac{1}{1-1/e}$,}
providing us with a guideline of the smallest possible gap that we possibly could expect with online policies.

\begin{theorem}\label{thm:online-opt-EXP}
  Under Poisson
  \revrev{arrivals}{requests}
  we have
  \begin{align}
    \revrev{}{\frac{C_{opt}^{online}}{C_{opt}^{offline}} =}
    \frac{C_{opt}^{static}}{C_{opt}^{offline}} \le
    \revrev{\frac{e}{e-1} = \frac{1}{1-1/e}.}{\frac{1}{1-1/e}.}
  \end{align}
\end{theorem}

\begin{proof}
  \revrev{First,}{The first equality comes directly from Theorem~\ref{thm:non-decreasing-hazard-rate}.  Now,}
  let us identify the
  \revtwo{arrival}{request}
  rate where the ratio between
  \revrev{$C_{opt}^{online}$}{$C_{opt}^{static}$}
  and $C_{opt}^{offline}$ is the greatest.
  This can be shown by first noting that
  $\frac{d}{d\lambda}(\frac{\lambda R}{1-e^{-\lambda R}}) = \frac{Re^{\lambda R}(e^{\lambda R}-\lambda R - 1)}{(e^{\lambda R}-1)^2} \ge 0$
  and that $\frac{d}{d\lambda}(\frac{1}{1-e^{-\lambda R}}) = - \frac{Re^{\lambda R}}{(e^{\lambda R}-1)^2} \le 0$.
  Therefore,
  the maximum ratio $\frac{C_{opt}^{online}}{C_{opt}^{offline}}$ is obtained when $\lambda = \frac{1}{R}$.
  Insertion into the expressions (\ref{eqn:offline-opt-general-EXP})
  and (\ref{eqn:online-opt-poisson}) and taking the ratio completes the proof.
  \end{proof}

\subsection{Erlang with known $k$ and $\lambda$}

\begin{theorem}
  Under Erlang inter-request times, we have
  \begin{align}
    \revrev{}{\frac{C_{opt}^{online}}{C_{opt}^{offline}} =}
    \frac{C_{opt}^{static}}{C_{opt}^{offline}} \le \frac{1}{1-e^{-k}\frac{k^k}{k!}}.
  \end{align}
\end{theorem}

\begin{proof}
  \revrev{}{Similarly as for a Poisson request process, the optimal {\em static baseline} policy has cost equal to
    the minimum of $R$ divided by the average inter-request time (with Erlang inter-request times,
    equal to $k/\lambda$), and 1.}
  Consider first the low-rate ratio,
between never caching (at cost $\min[\frac{\lambda}{k} R,1]$=$\frac{\lambda}{k}R$) and {\em optimal offline} (equation (\ref{eqn:offline-opt-general-ERLANG})):
\begin{align}
  \label{eqn:FG}
  \frac{F}{G} = \frac{ \frac{\lambda}{k} R }{1 - \frac{e^{-\lambda R}}{k} \sum_{m=1}^{k} \sum_{n=0}^{m-1} \frac{(\lambda R)^n}{n!} },
\end{align}
where we have used $F$ and $G$ to denote the nominator and denominator.
\revfour{Now, taking}{Taking}
the derivative with respect to $\lambda$ we obtain:
\begin{align}
  \frac{d}{d\lambda} (\frac{F}{G}) & = \frac{1}{G^2}(\frac{dF}{d\lambda}G - F \frac{dG}{d\lambda}) \nonumber \\
  & = \frac{1}{G^2} \left( \frac{R}{k} \left(1 - \frac{e^{-\lambda R}}{k} \sum_{m=1}^{k}\sum_{n=0}^{m-1} \frac{(\lambda R)^n}{n!} \right) - \frac{\lambda R}{k}\frac{R}{k}e^{-\lambda R} \sum_{n=0}^{k-1} \frac{(\lambda R)^n}{n!}  \right) \nonumber \\
  & = \frac{1}{G^2} \left( \frac{R}{k} - \frac{R}{k} e^{-\lambda R} \sum_{n=0}^{k-1} \frac{(\lambda R)^n}{n!} + \frac{\lambda R^2}{k^2} e^{-\lambda R} \left( \frac{(\lambda R)^{k-1}}{(k-1)!} \right) \right) \nonumber \\ 
  & =  \frac{1}{G^2} \left( \frac{R}{k} - \frac{R}{k}e^{-\lambda R} \sum_{n=0}^k \frac{(\lambda R)^n}{n!}  \right).
\end{align}
\revrev{where we have used that
$\frac{dF}{d\lambda} = R/k$,
$\frac{dG}{d\lambda} = \frac{R}{k} e^{-\lambda R} \sum_{m=1}^{k} \sum_{n=0}^{m-1} \frac{(\lambda R)^n}{n!} - \frac{1}{k}  e^{-\lambda R}  \sum_{m=1}^{k} \sum_{n=1}^{m-1} \frac{(\lambda R)^{n-1}}{(n-1)!}
= \frac{R}{k} e^{-\lambda R} \sum_{m=1}^{k} \left( \sum_{n=0}^{m-1} \frac{(\lambda R)^n}{n!} - \sum_{n=0}^{m-2} \frac{(\lambda R)^n}{n!} \right)
= \frac{R}{k} e^{-\lambda R} \sum_{m=1}^{k} \frac{(\lambda R)^{m-1}}{(m-1)!}
= \frac{R}{k} e^{-\lambda R} \sum_{n=0}^{k-1} \frac{(\lambda R)^n}{n!}$,
and (for the third step) the following rewrite:
$\sum_{m=1}^{k}\sum_{n=0}^{m-1} \frac{(\lambda R)^n}{n!} = \sum_{n=0}^{k-1} (k-n) \frac{(\lambda R)^n}{n!}
= k \sum_{n=0}^{k-1} \frac{(\lambda R)^n}{n!} - \lambda R \sum_{n=0}^{k-2} \frac{(\lambda R)^n}{n!}$.
\revfour{Now, since}{Since}
the last sum is no greater than $e^{\lambda R}$ (i.e., $\sum_{n=0}^k \frac{(\lambda R)^n}{n!} \le e^{\lambda R}$),}{Now,
  since $\sum_{n=0}^k \frac{(\lambda R)^n}{n!} \le e^{\lambda R}$,}
we have that $\frac{d}{d\lambda} (\frac{F}{G}) \ge 0$. 
This shows that the worst case
\revfour{ratio}{ratio when $\frac{\lambda}{k} R \le 1$}
is observed when $\lambda = \frac{k}{R}$.
Insertion into
\revfour{the expression}{expression (\ref{eqn:FG})}
gives
\revfour{us the}{the}
bound:
\begin{align}
  \frac{F}{G} & = \frac{1}{1 - \frac{e^{-k}}{k} \sum_{m=1}^{k}\sum_{n=0}^{m-1}\frac{k^n}{n!}}
  = \frac{1}{1 - \frac{e^{-k} k^k}{k!}}.
\end{align}
Similarly, when $\frac{\lambda}{k} R > 1$ (and $\min[\frac{\lambda}{k} R,1] = 1$),
\revfour{we note}{it is straightforward to show}
that $\frac{d}{d\lambda} (\frac{F}{G}) \le 0$,
and the worst case therefore again occurs when $\lambda = \frac{k}{R}$. 
\end{proof}

Note that the Erlang
\revrev{bound}{competitive ratio} approaches 1 as $k \rightarrow \infty$.

\subsection{Deterministic with known $a$}

\begin{theorem}\label{thm:online-deterministic}
  Under deterministic inter-request times, we have
  \begin{align}
    \revrev{}{\frac{C_{opt}^{online}}{C_{opt}^{offline}} =}
    \frac{C_{opt}^{static}}{C_{opt}^{offline}} = 1.
  \end{align}
\end{theorem}

\begin{proof}
  Since knowledge of the (constant) inter-request time
  is equivalent to knowledge of the entire request sequence,
  \revfour{the policies}{the optimal {\em static baseline}
    (same as \revfour{online optimal}{{\em online optimal}}) and \revfour{offline optimal}{{\em offline optimal}} policies}
are identical.  When $a \le \frac{1}{R}$,
  both policies keeps the object cached all the time,
  and when $\frac{1}{R} < a$ neither policy caches the object.
\end{proof}

\subsection{Pareto with known $\alpha$ and $t_m$}

\begin{theorem}\label{thm:static-pareto}
With Pareto inter-request times,
the worst-case cost ratio for the optimal {\em static baseline}
is unbounded.  In particular,
  \begin{align}
    \frac{C_{opt}^{static}}{C_{opt}^{offline}} \rightarrow \infty
  \end{align}
  when
  \revrev{$\alpha^{*} = \frac{1}{1-\frac{t_m}{R}}$}{$\alpha = \frac{1}{1-\frac{t_m}{R}}$}
  and $\frac{t_m}{R} \rightarrow 0+$.
\end{theorem}

\begin{proof}

  Assuming Pareto distributed inter-request times,
  \revfour{this}{the optimal {\em static baseline}}
  policy has cost:
  \begin{align}\label{eqn:static-pareto}
    C_{opt}^{static} = \min[\frac{\alpha-1}{\alpha} \frac{R}{t_m},1].
  \end{align}

  \revfour{Consider the ratio of $\frac{\alpha-1}{\alpha} \frac{R}{t_m}$
    \revrev{(when this is less than 1)}{(when this is at most 1)}}{Assume first that $\frac{\alpha-1}{\alpha} \frac{R}{t_m} \le 1$,
    and consider the ratio of this quantity}
  and the
  \revrev{offline bound.}{offline bound for $t_m \le R$.  (In the case of $t_m > R$, the cost ratio is 1.)}
  This
  \revfour{function}{ratio}
  has a non-negative derivative:
  \begin{align}
    \frac{d}{d\alpha} \left( \frac{\frac{\alpha-1}{\alpha} \frac{R}{t_m}}{1-\frac{1}{\alpha} \left( \frac{t_m}{R} \right)^{\alpha-1} } \right) \ge 0.
  \end{align}
  Now, let
  \revrev{$x=\frac{t_m}{R}$ and consider the}{$x=\frac{t_m}{R}$. The}
  maximum
  value of $\alpha$ for which
  \revrev{$\frac{\alpha-1}{\alpha} \frac{R}{t_m} < 1$
    when $\alpha^{*} = \frac{1}{1-\frac{t_m}{R}} = \frac{1}{1-x}$.}{$\frac{\alpha-1}{\alpha} \frac{R}{t_m} \le 1$
    is given by $\frac{1}{1-x}$.}
  For this point, the ratio is:
  \begin{align}
    \frac{C_{opt}^{static}}{C_{opt}^{offline}} \le \frac{1}{1-(1-x)x^{x/(1-x)}}.
  \end{align}
  \revfour{Taking the derivative of this function:}{Taking the derivative of this function with respect to $x$,
    it can be seen that the ratio is non-increasing in $x$:}
  \begin{align}
    \frac{d}{dx} \left( \frac{1}{1-(1-x)x^{x/(1-x)}} \right)  = \frac{x^{x/(1-x)} \ln x}{(1-x) (1-(1-x)x^{x/(1-x)})^2} \le 0,
  \end{align}
  \revfour{we see that}{and so}
  the largest ratio occurs when
  \revrev{$\alpha^{*} = \frac{1}{1-x}$}{$\alpha = \frac{1}{1-x}$}
  and $x \rightarrow 0+$.
  \revfour{In this limit,}{In this case,}
  $\frac{1}{1-(1-x)x^{x(1-x)}} \rightarrow \infty$
  \revfour{as $x \rightarrow 0+$, and}{and}
  the worst-case ratio is therefore unbounded.
\end{proof}

The above result illustrates the importance of using a bounded
TTL value to remove stale objects from the cache.

\section{Steady-state: Insertion Policies}\label{sec:general-policies}

\revfour{}{We next derive expressions for the delivery costs of the {\em cache on $M^{th}$ request}
  policies outlined in Section~\ref{sec:policies}.
  We again assume that inter-request times are independent and identically distributed,
  with a general inter-request time distribution $f(t)$.  In Section~\ref{sec:results},
  we then use these results to derive explicit expression for the four example distributions
  considered in this paper.  Using these general results, it is of course straightforward to derive explicit
  expressions
for other distributions also.}

\subsection{Always on $1^{st}$ ($T$):}

To derive the average cost per time unit,
we consider an arbitrary renewal period that includes
both a ``busy period'' (during which the object is in the cache)
and an ``off period'' (during which the object is not in the cache).
\revtwo{Now, under our distribution assumptions, the}{The}
average cost can be calculated
as the total expected cost accumulated over such a
\revtwo{combined}{renewal}
period (i.e., $R$ plus the time the object stays in the cache)
divided by the expected duration of the renewal period
\revfour{(i.e., $E[a_i]$ + the time the object stays in the cache).}{(i.e., the expected time from when the object
  is added to the cache until it is removed, plus the expected time from when the object is removed from the cache until its next request).}
Therefore,
\begin{align}
  \label{eqn:C_k1_v1}
  C^{always}_{M=1,T} & = \frac{R + E[\Theta]}{E[\Delta_1] + E[\Theta]},
\end{align}
where $E[\Theta]$ is the expected
\revtwo{duration}{time}
that the object is in the cache and
\begin{align}\label{eqn:Delta_1}
  E[\Delta_{1}] & = E[a_i | a_i > T] - T = \frac{1}{1-F(T)}\left(E[a_i] + \int_0^T F(t) \textrm{dt} - T \right),
\end{align}
is the expected time until the next request, given that the object was just removed from the cache.
To derive an expression for $E[\Theta]$,
we identify and solve the following recurrence:
\begin{align}
  E[\Theta] & = (1-F(T)) T + F(T) (E[a_i | a_i < T] +  E[\Theta]),
\end{align}
where $E[a_i | a_i < T]$ is the expected time between two
\revfour{}{consecutive}
requests,
given that the inter-request time between the two requests is less than $T$.
\revfour{Here, we have used that the object is refreshed only
when there is a request within time period $T$
(which happens with a probability $F(T)$)
and that the expected
remaining time that the object will remain in the cache only depends on
the most recent request (not on how many prior requests there have been while in the cache).}{This recurrence follows from
  the fact that the object is removed from the cache after time $T$ if there have been no new requests for it
  (probability $1 - F(T)$), and that otherwise (probability $F(T)$) the object's lifetime in the cache is refreshed
  at the time of the first new request.}
Now, solving for $E[\Theta]$ we obtain:
\begin{align}\label{eqn:theta}
  E[\Theta] & = T + \frac{F(T)}{1 - F(T)} E[a_i | a_i < T]
  = \frac{1}{1-F(T)}\left( T - \int_0^{T}F(t)\textrm{dt} \right) ,
\end{align}
where we have used equation (\ref{eqn:EltX}) in the second step.
Insertion of equations (\ref{eqn:Delta_1}) and (\ref{eqn:theta}) into equation (\ref{eqn:C_k1_v1})
\revfour{now gives us:}{gives:}
\begin{align}
  \label{eqn:C_k1_v2}
  C^{always}_{M=1,T}
  & = \frac{(1-F(T))R + T - \int_0^T F(t) \textrm{dt}}{E[a_i]}. 
\end{align}

\subsection{Always on $M^{th}$ ($M,T$)}

\revfour{Similar to}{As}
for the {\em always on $1^{st}$} policy,
\revfour{with}{for}
the {\em always on $M^{th}$} policy
we can analyze an arbitrary renewal period.
\revfour{Since each such renewal period would see $M-1$ more requests (that are not served from the cache),
  each period would be $(M-1)E[a_i]$ longer and have a cost $(M-1)R$ higher.}{Since $M$ requests are needed for an uncached object
  to be added to the cache, the off period is $(M-1) E[a_i]$ longer than for {\em always on 1$^{st}$},
  and the total expected cost over a renewal period is $(M-1) R$ higher.}
\revtwo{Otherwise, the}{The}
time that the object stays in the cache is the same as for the {\em always on $1^{st}$} policy.
\revtwo{Given these observations, the average cost per time unit for the {\em always on $M^{th}$} policy
can be calculated as:}{These observations yield:} 
\begin{align}\label{eqn:C_k2always_v2}
  C^{always}_{M,T} & = \frac{MR + E[\Theta]}{E[\Delta_1] + (M-1)E[a_i] + E[\Theta]}
  \nonumber\\
  &
  = \frac{(1-F(T))MR + T - \int_0^T F(t) \textrm{dt}}{(M - F(T))E[a_i]}.
\end{align}

\subsection{Single-window on $M^{th}$ ($M,T$)}

The average cost per time unit
can be calculated using the formula:
\begin{align}
  \label{eqn:C_M_v1}
  C^{window}_{M,T} = \frac{E[N_{M}]R+E[\Theta]}{E[\Delta_M]+E[\Theta]},
\end{align}
where $E[N_{M}]$ is the expected number of requests needed before the object re-enters the
\revtwo{cache and}{cache,}
$E[\Delta_M]$ is the expected time duration that the object is not in the cache during a renewal
\revtwo{period. Again, the expected time duration $E[\Theta]$ that
  the object is in the cache during such a renewal period}{period, and $E[\Theta]$}
is the same as for the prior two policies analyzed.

To obtain $E[\Delta_{M}]$, we identify the following recurrence:
\begin{align}
  E[\Delta_{M}] & = E[\Delta_{M-1}] + F(T) E[a_i| a_i \le T] + (1-F(T))(T+E[\Delta_{M}]).
\end{align}
Solving for $E[\Delta_{M}]$ and using equation (\ref{eqn:Delta_1}) for
\revfour{}{the base case of the recurrence}
$E[\Delta_1]$,
\revtwo{we obtain the following:}{we obtain:}
\begin{align}\label{eqn:delta_M}
  E[\Delta_{M}] & = \frac{1}{F(T)}\left( E[\Delta_{M-1}] + T - \int_0^T F(t) \textrm{dt}\right) \nonumber\\
  & = \frac{1}{1-F(T)} \left( \frac{E[a_i]}{F(T)^{M-1}} + \int_0^T F(t) \textrm{dt} - T \right).
\end{align}

Similarly, to obtain $E[N_{M}]$, we identify the following recurrence:
\begin{align}
  E[N_{M}] & = E[N_{M-1}] + F(T) + (1-F(T))E[N_{M}].
\end{align}
Solving for $E[N_{M}]$ and recognizing that $E[N_{1}] = 1$,
\revtwo{we obtain the following:}{we obtain:}
\begin{align}\label{eqn:N_M}
  E[N_{M}] & = 1 + \frac{E[N_{M-1}]}{F(T)} = \sum_{i=0}^{M-1} \frac{1}{F(T)^i}.
\end{align}
Inserting equations (\ref{eqn:theta}), (\ref{eqn:delta_M}) and (\ref{eqn:N_M}) into equation (\ref{eqn:C_M_v1}) we obtain:
\begin{align}
  \label{eqn:C_M_v2}
  C^{window}_{M,T}
  & = \frac{ (1-F(T)) \sum_{i=0}^{M-1} \frac{1}{F(T)^i} R + \left( T - \int_0^{T}F(t)\textrm{dt} \right)}{\frac{E[a_i]}{F(T)^{M-1}}}.
\end{align}

\subsection{Dual-window on $2^{nd}$ ($W,T$)}

\revtwo{For the analysis of this policy,
we assume that the object always must see at least two requests within $W$
of each other {\em after} the object was removed from the cache.
This is always the case when $W \le T$ (considered in this paper).
Now, the average cost per time unit
can be calculated using the formula:}{Note that since we assume $W \le T$,
  the two requests within $W$ of each other that are required for an evicted object
  to be cached again must occur after the object eviction.  The average cost per time unit is given by}
\begin{align}
  \label{eqn:C_k2_v1}
  C^{window}_{M=2,W, T} = \frac{E[N_{2}]R+E[\Theta]}{E[\Delta_2]+E[\Theta]},
\end{align}
where $E[N_{2}]$ is the expected number of requests needed before the object re-enters the
\revtwo{cache and}{cache,}
$E[\Delta_2]$ is the expected time duration that the object is not in the cache during a renewal
\revtwo{period. Again, the expected time duration $E[\Theta]$ that the object is
  in the cache during such a renewal period}{period, and $E[\Theta]$}
is the same as for the prior
\revtwo{two policies analyzed.}{policies.}
\revtwo{Given the above assumptions, the
expected time duration that the object is out of the cache}{Here, $E[\Delta_2]$}
can be expressed as
\begin{align}
  \label{eqn:Delta_G_v1}
  E[\Delta_2] & = E[a_i-T | a_i > T] + E[\delta] = E[a_i | a_i > T] - T + E[\delta]\nonumber\\
  & = \frac{1}{1-F(T)} \left(E[a_i] - T F(T) + \int_0^T F(t) \textrm{dt} \right) - T + E[\delta],
\end{align}
where $E[\delta]$ can be expressed using the following recurrence:
\revfour{\begin{align}
  E[\delta] & = P(a_i \le W) E[a_i | a_i \le W] + (1-P(a_i \le W))\left(E[a_i | a_i > W] + E[\delta] \right) \nonumber\\
  & = F(W) E[a_i | a_i \le W] + (1-F(W)) \left(E[a_i | a_i > W] + E[\delta] \right). 
\end{align}}{\begin{align}
    E[\delta]
    & = F(W) E[a_i | a_i \le W] + (1-F(W)) \left(E[a_i | a_i > W] + E[\delta] \right).
\end{align}}

Solving for $E[\delta]$, we obtain:
\revfour{
$E[\delta] = E[a_i | a_i \le W] + \frac{1-F(W)}{F(W)} E[a_i | a_i > W]
= \frac{1}{F(W)} \left( F(W) E[a_i | a_i \le W] + (1-F(W)) E[a_i | a_i > W] \right)
= \frac{1}{F(W)}E[a_i]$.}{
  \begin{align}\label{eqn:delta}
  E[\delta] & = E[a_i | a_i \le W] + \frac{1-F(W)}{F(W)} E[a_i | a_i > W] \nonumber\\
  & = \frac{1}{F(W)} \left( F(W) E[a_i | a_i \le W] + (1-F(W)) E[a_i | a_i > W] \right) \nonumber\\
  & = \frac{1}{F(W)}E[a_i].
  \end{align}}
Insertion into equation (\ref{eqn:Delta_G_v1}) then gives:
\begin{align}
  \label{eqn:Delta_G_v2}
  E[\Delta_2] & = \frac{1}{1-F(T)} \left(E[a_i] - T F(T) + \int_0^T F(t) \textrm{dt} \right) - T + \frac{E[a_i]}{F(W)}.
\end{align}

Similarly,
\revtwo{the expected number of requests $E[M_{2}]$
needed before the object re-enters the cache
can be expressed as:}{the expected number of requests $E[N_2]$ needed before the object re-enters the cache can be expressed as}
\begin{align}
  \label{eqn:m_G_v1}
  E[N_{2}] =  1 + E[m],
\end{align}
where $E[m]$ can be expressed using the following recurrence:
\revfour{$ E[m] = F(W) \cdot 1 + (1-F(W)) \cdot (1+E[m])$.}{$E[m] = F(W) + (1-F(W))(1+E[m])$.}
Solving for $E[m]$, we obtain:
 $ E[m] = 1 + \frac{1-F(W)}{F(W)}$.
Insertion into equation (\ref{eqn:m_G_v1}) then gives:
\begin{align}
  \label{eqn:m_G_v2}
  E[N_{2}] = 2 + \frac{1-F(W)}{F(W)}.
\end{align}

Finally,
\revtwo{inserting}{substituting}
equations (\ref{eqn:Delta_G_v2}), (\ref{eqn:m_G_v2}) and (\ref{eqn:theta})
into equation (\ref{eqn:C_k2_v1}),
\revtwo{and multiplying both the numerator and denominator with $(1-F(T))$, we obtain:}{and simplifying, yields}
\begin{align}
  \label{eqn:C_k2_v2}
  C^{window}_{M=2,W, T} & = \frac{(1-F(T))\left(2 + \frac{1-F(W)}{F(W)}\right)R+\left( T - \int_0^{T}F(t)\textrm{dt} \right) }{E[a_i] (1 + \frac{1-F(T)}{F(W)})}.
\end{align}

\section{Results for example distributions}\label{sec:results}

We next present explicit expressions for
\revtwo{each {\em cache on $M^{th}$ request} policy}{the policies}
considered in this paper
for four different distributions: exponential, Erlang, deterministic, and Pareto.
Table~\ref{tab:summary-costs} summarizes these results.
For derivations of
\revfour{the {\em optimal offline} policies (top row) and the {\em static baselines} (second row),
  optimized for when request rates are known,}{the {\em optimal offline} results (top row),
  and the {\em static baseline} results that assume a known inter-request time distribution (second row),}
we refer to Sections~\ref{sec:general} and~\ref{sec:baselin}, respectively.
We next present and discuss results for each
\revfour{distribution of consideration.}{considered distribution.}

\begin{table*}[t]
  \caption{Summary of costs for different distributions and insertion policies.
    To make room, for Erlang, we simplified expressions using
    $F(t) = 1 - \sum_{n=0}^{k-1}\frac{1}{n!}e^{-\lambda t}(\lambda t)^n$ and
    $\Phi(T)=\frac{e^{-\lambda T}}{\lambda} \sum_{m=1}^{k} \sum_{n=0}^{m-1} \frac{(\lambda T)^n}{n!}$.}
  \label{tab:summary-costs}
  \vspace{-0pt}
{\tiny
\begin{tabular}{|l|c|c|p{1.5cm}|c|}
\hline
 Policy & Exponential & Erlang & Deterministic & Pareto \\ \hline
Offline         & $1-e^{-\lambda R}$
& $1 - \frac{\lambda}{k} \Phi(R)$
& $\: \min[\frac{R}{a},1]$
& $\begin{array}{ll} 1 - \frac{1}{\alpha} \left( \frac{t_m}{R} \right)^{\alpha - 1}, & \textrm{if}~ t_m \le R\\  \frac{R(\alpha-1)}{\alpha t_m}, & \textrm{if}~ R < t_m\\ \end{array}$
\\\hline

Baseline        & $\min[\lambda R, 1]$
& $\min[\frac{\lambda}{k} R, 1]$
& $\: \min[\frac{R}{a}, 1]$
& $\min[\frac{\alpha-1}{\alpha}\frac{R}{t_m},1]$
\\\hline

Always $1^{st}$
& $1-e^{-\lambda T} + \lambda R e^{-\lambda T}$
& $(1-F(T))\frac{\lambda}{k}R + (1 - \frac{\lambda}{k}\Phi(T))$
& $\begin{array}{l} 1, \qquad \textrm{if}~ a \le T\\ \frac{R+T}{a}, \: \textrm{if}~T < a\\  \end{array}$
& $\begin{array}{ll} \frac{\alpha-1}{\alpha}\left(\frac{t_m}{T}\right)^{\alpha}\frac{R}{t_m} + \left(1 - \frac{1}{\alpha}\left(\frac{t_m}{T}\right)^{\alpha-1}\right), & \textrm{if}~ t_m \le T\\    \frac{(R+T)(\alpha - 1)}{\alpha t_m}, & \textrm{if}~ T < t_m\\ \end{array}$
\\\hline

Always $2^{nd}$
& $\frac{1-e^{-\lambda T} + 2 \lambda Re^{-\lambda T}}{1 + e^{-\lambda T}}$
& $\frac{(1-F(T))\frac{\lambda}{k}2R + (1 - \frac{\lambda}{k}\Phi(T))}{2 - F(T)}$
& $\begin{array}{l} 1, \qquad \textrm{if}~ a \le T\\  \frac{2R+T}{2a}, \: \textrm{if}~ T < a\\ \end{array}$
& $\begin{array}{ll} \frac{\frac{\alpha-1}{\alpha}\left(\frac{t_m}{T}\right)^{\alpha}\frac{2R}{t_m} + \left(1 - \frac{1}{\alpha}\left(\frac{t_m}{T}\right)^{\alpha-1}\right)}{1+\left(\frac{t_m}{T}\right)^{\alpha}}, & \textrm{if}~ t_m \le T\\  \frac{2R+T}{2}\frac{\alpha - 1}{\alpha t_m}, & \textrm{if}~ T < t_m\\ \end{array}$
\\\hline

Single $M^{th}$ &
$\lambda e^{-\lambda T} \sum_{i=0}^{M-1}(1 -  e^{-\lambda T})^i R + \left(1-e^{-\lambda T} \right)^M$
& $\begin{array}{l}(1-F(T)) \frac{\lambda}{k} \sum_{i=0}^{M-1} F(T)^i R \\ \qquad + \left( 1 - \frac{\lambda}{k}\Phi(T)\right)F(T)^{M-1} \end{array}$
& $\begin{array}{l} 1, \qquad \textrm{if} a \le T\\ \frac{R}{a}, \qquad  \textrm{if}~T < a\\  \end{array}$
& $\begin{array}{ll}  \frac{\alpha-1}{\alpha}\left(\frac{t_m}{T}\right)^{\alpha} \sum_{i=0}^{M-1} (1 - \left(\frac{t_m}{T}\right)^{\alpha})^i \frac{R}{T} \\ \qquad \qquad + \left(1 - \frac{1}{\alpha}\left(\frac{t_m}{T}\right)^{\alpha-1} \right)(1 - \left(\frac{t_m}{T}\right)^{\alpha})^{M-1}, & \textrm{if}~ t_m \le T\\  \frac{R(\alpha - 1)}{\alpha t_m}, & \textrm{if}~ T < t_m\\ \end{array}$
\\\hline

Dual $2^{nd}$ & $\frac{\lambda R e^{-\lambda T} \left(2 - e^{-\lambda W} \right) + \left( 1-e^{-\lambda T} \right)\left(1-e^{-\lambda W}\right) }{ 1 - e^{-\lambda W} + e^{-\lambda T}}$
& $\frac{(1-F(T))\left(2 + \frac{1-F(W)}{F(W)}\right)R +\left( \frac{k}{\lambda} - \Phi(T) \right) }{\frac{k}{\lambda} (1 + \frac{1-F(T)}{F(W)})}$
& $ \begin{array}{l} 1, \:  \textrm{if}~a < W \le T \\ \frac{R}{a}, \:  \textrm{if} W < a\\  \end{array}$
& $\begin{array}{ll}  \frac{(\alpha-1)\left(\frac{t_m}{T}\right)^{\alpha}(2-\left(\frac{t_m}{W}\right)^{\alpha})R + (1-\left(\frac{t_m}{W}\right)^{\alpha})(t_m \alpha - T\left(\frac{t_m}{T}\right)^{\alpha})}{\alpha t_m (1 - \left(\frac{t_m}{W}\right)^{\alpha} + \left(\frac{t_m}{T}\right)^{\alpha})}, & \textrm{if}~ t_m \le W\\  \frac{R(\alpha - 1)}{\alpha t_m}, & \textrm{if}~ W < t_m\\ \end{array}$
\\\hline
\end{tabular}}
\vspace{-4pt}
\end{table*}

{\bf Exponential:}  The results for the
\revfour{}{four}
insertion policies are obtained
\revfour{by inserting equations (\ref{eqn:exponential-A}) and (\ref{eqn:exponential-B})
  into equations (\ref{eqn:C_k1_v2}), (\ref{eqn:C_k2always_v2}), (\ref{eqn:C_M_v2}), (\ref{eqn:C_k2_v2}),}{by using equations
  (\ref{eqn:exponential-A}) and (\ref{eqn:exponential-B}) to substitute for $E[a_i]$, $F(t)$ and the integral of $F(t)$
  in equations (\ref{eqn:C_k1_v2}), (\ref{eqn:C_k2always_v2}), (\ref{eqn:C_M_v2}), (\ref{eqn:C_k2_v2}),}
and then simplifying the expressions.
For example, for the {\em always on $1^{st}$} policy,
\revfour{we insert (\ref{eqn:exponential-A}) and (\ref{eqn:exponential-B})
into equation (\ref{eqn:C_k1_v2}):}{using equations (\ref{eqn:exponential-A}) and (\ref{eqn:exponential-B}) to substitute into equation (\ref{eqn:C_k1_v2}) yields:}
\begin{align}\label{eqn:C_k1_v2-EXP}
  C^{always}_{M=1,T} = 1-e^{-\lambda T} + \lambda R e^{-\lambda T}.
\end{align}
\revfour{We note that the derivative (with respect to $T$):}{Note that the derivative of the cost with respect to $T$, as given by}
\begin{align}
  \frac{d}{dT}\left(C^{always}_{M=1,T}\right) = (\lambda - R \lambda^2) e^{-\lambda T},
\end{align}
is negative for $\lambda < R$ and positive for $R < \lambda$.
Therefore, for the (unrealistic) case that
\revtwo{arrival}{request}
rates are known,
it would be optimal to never cache (i.e., use $T=0$)
for
\revfour{files}{file objects}
with $\lambda \le R$ and never empty the cache
(i.e., $T \rightarrow \infty$) when $R < \lambda$.
For these two extreme cases,
\revtwo{we have an}{the}
average (expected) cost
\revtwo{of}{is}
$\lambda R$ and 1, respectively.
Taking the better of these corresponds to our (optimal) {\em static baseline} policy.
\revtwo{Of course, in practice the
\revtwo{arrival}{request}
rate $\lambda$ is not known
and an intermediate $T$ may therefore be desirable to avoid worst-case cost ratios,
relative to {\em optimal offline}, to sky-rocket, for example.}{}

\revtwo{As with the optimal {\em static baseline} policy (assuming known
\revtwo{arrival}{request}
rates),
also this policy can have an unbounded worst-case cost ratio compared to both
optimal offline and online policies if setting $T$ to large or to small.}{With unknown request rate,
  an intermediate value of $T$ is needed to avoid unbounded worst-case cost ratios.}
\revfour{For example, with $T$=$0$ we get unbounded performance penalty when $\lambda$$\rightarrow$$\infty$
and with $T$$\rightarrow$$\infty$ we get
\revtwo{bounded}{unbounded}
performance penalty when $\lambda$$\rightarrow$$0$.
Motivated}{Motivated}
by our worst-case analysis for arbitrary request distributions (Section~\ref{sec:worse-case}),
\revfour{we instead}{we}
focus our attention
\revtwo{on the case when $T=R$;
both for this policy and for all other policies of consideration.}{on policies using $T$=$R$.}
\revtwo{Again, note that {\em always on $1^{st}$} has a
worst-case bound of 2 (Theorem~\ref{thm:always-1st}).
Interestingly,}{Interestingly,}
taking the ratio of equations (\ref{eqn:C_k1_v2-EXP}) and
\revtwo{(\ref{eqn:offline-opt-general-EXP}) and taking the limit $\lambda \rightarrow 0$,}{(\ref{eqn:offline-opt-general-EXP}),}
\revfour{it is easy to see that}{it can be seen that}
\revtwo{this worst-case bound}{the worst-case bound of 2 shown in Theorem~\ref{thm:always-1st} for {\em always on $1^{st}$}}
  is achieved
  \revtwo{for the low-request rate region:}{with exponential inter-request times as $\lambda$$\rightarrow$$0$:}
\begin{align}
  \lim_{\lambda \rightarrow 0} \frac{C^{always}_{M=1,T=R}}{C^{offline}_{opt}} & = \lim_{\lambda \rightarrow 0} \frac{1-e^{-\lambda R} + \lambda R e^{-\lambda R}}{1-e^{-\lambda R}}
  = \lim_{\lambda \rightarrow 0} \frac{\lambda R + \lambda R}{\lambda R} = 2.
\end{align}
\revtwo{Here, we have used Taylor expansion of both nominator and denominator in the second step.
Using the same technique,}{Similarly,}
\revfour{it is easy}{it is straightforward}
to show that the
\revtwo{ratio of}{cost ratio, with exponential inter-request times and $\lambda$$\rightarrow$$0$, for}
       {\em always on $M^{th}$} is $\frac{M+1}{M}$ and
       \revtwo{the}{that for}
              {\em single-window on $M^{th}$} is 1 when $M$$\ge$$2$.  This is encouraging,
since it shows that {\em single-window on $M^{th}$} in practice may significantly
outperform {\em always on $1^{st}$}, despite a looser worst-case
\revtwo{bound, as well as the other {\em always on $M^{th}$}.}{bound.}

In fact, using {\em single-window on $2^{nd}$} with the optimal worst-case analysis setting of
\revtwo{$W=T=R$,}{$T=R$,}
the largest cost ratio (across the full range of request rates)
is only slightly higher than for the
\revtwo{(optimal)}{(optimal assuming known request rate)}
       {\em static baseline},
which has a peak ratio of $\frac{1}{1-1/e} \approx 1.582$ (when $\lambda R = 1$),
\revfour{as per}{as shown in}
Theorem~\ref{thm:online-opt-EXP}.
\revtwo{For example, taking}{This can be seen by taking}
the ratio of the cost functions of {\em single-window on $2^{nd}$} and the {\em offline optimal}:
\begin{align}
  \frac{\lambda R e^{-\lambda R} (2-e^{-\lambda R}) + (1-e^{-\lambda R})^2}{1-e^{-\lambda R}},
\end{align}
\revtwo{two extreme points can be identified:}{and identifying the two extreme points:}
$\lambda R = 0$ and $\lambda R  \approx 1.05236$ (numerically).
When $\lambda \rightarrow 0$ the ratio is 1 and when $\lambda R  \approx 1.05236$ the ratio is 1.588.
\revtwo{It is, however, worth remembering that the (optimal) {\em static baseline} policy
can perform {\em much} worse if optimizing with regards to the wrong request rate.
In contrast, {\em single-window on $2^{nd}$} with $W=T=R$
performs consistently good throughout the full parameter range without
having to tune any parameter based on the current request rate.}{}

Figure~\ref{fig:exponential} summarizes the performance of the different {\em cache on $M^{th}$} policies.
Here, we have used
\revfour{\revtwo{$W=T=R$ for all policies.}{$W=T=R$.}}{$W=T=R$,
  and on the x-axis vary the ```normalized average request rate'' as given by the average number
  of requests within a window of $W = T$ time units.
  For example, an x-axis value of 1 corresponds to an average request rate of $\lambda = 1/T = 1/W = 1/R$.}
\revfour{In addition to clearly observing}{Note}
that the window-based policies significantly outperform the
{\em always on $M^{th}$ policies},
\revfour{we note}{and}
that {\em single-window on $2^{nd}$} with $T=R$
achieves good performance throughout,
as it closely tracks
\revtwo{{\em static baseline}. Interesting, for the exponential case (as well as for deterministic),
it is possible to show that the optimal setting for the different
{\em cache on $M^{th}$ request} policies (both of type ``always'' and ``window'')
is either $T = 0$ or $T \rightarrow \infty$ when request rates are known.
For example, consider the derivative with respect of $T$ of the
simplest {\em always on $1^{st}$} policy:
\begin{align}
  \frac{dC^{k=1}_{T}}{dT} = (\lambda - R \lambda^2) e^{-\lambda T}.
\end{align}
This function is negative for $\lambda < R$ and positive for $R < \lambda$.
Therefore, in the (unrealistic) case that
\revtwo{arrival}{request}
rates are known,
it would be optimal to never cache (i.e., use $T=0$)
for files with $\lambda \le R$ and never empty the cache
(i.e., $T \rightarrow \infty$) when $R < \lambda$.
We conjecture that this also is true for the Erlang distribution,
but have yet to prove this for all cases.}{{\em static baseline},
  which bounds the optimal performance of any online policy when inter-request times
  are exponential (Theorem~\ref{thm:non-decreasing-hazard-rate}).}

Finally, comparing {\em single-window on $M^{th}$}
\revfour{with}{for}
$M=2$ and $M=4$,
we note that {\em single-window on $4^{th}$}
\revfour{actually tracks}{tracks}
the {\em static baseline}
even better up to
\revfour{the peak,}{the peak at $\lambda R = 1$,}
but then
\revfour{significantly overshoots for larger}{performs significantly worse for higher}
request rates.
With {\em single-window on $2^{nd}$}, there is a small but noticeable gap both before and after the peak.
However, the maximum difference is substantially smaller.

\begin{figure}[t]
  \centering
  \includegraphics[trim = 0mm 2mm 0mm 0mm, width=0.44\textwidth]{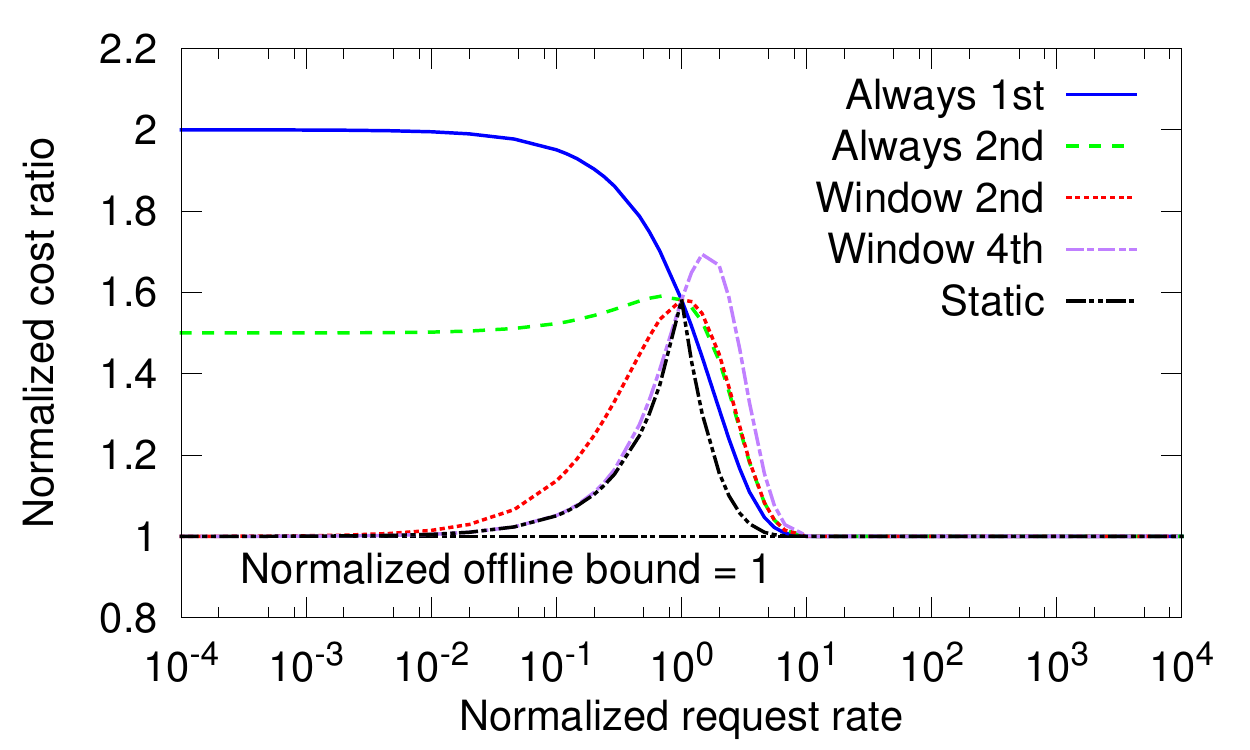}
  \vspace{-8pt}
  \caption{\revfour{Summary results using exponential distribution and $W=T=R$.}{Cost ratios for an exponential inter-request time distribution and $W=T=R$.}}
  \label{fig:exponential}
  \vspace{-6pt}
\end{figure}

{\bf Distributions with lower variability:}
Erlang results are obtained by
\revfour{inserting equations (\ref{eqn:erlang-A}) and (\ref{eqn:erlang-B}) into}{using equations
  (\ref{eqn:erlang-A}) and (\ref{eqn:erlang-B}) to substitute for $E[a_i]$, $F(t)$ and the integral of $F(t)$ in}
equations (\ref{eqn:C_k1_v2}), (\ref{eqn:C_k2always_v2}), (\ref{eqn:C_M_v2}), (\ref{eqn:C_k2_v2}),
and then
\revtwo{simplifying the expressions.}{simplifying.}
\revfour{Using similar techniques as for
\revtwo{the exponential distributions,}{the case of exponential inter-request times,}
it is possible to show that the low-rate performance (i.e., when $\lambda \rightarrow 0$)
is the same for all finite $k$ as for the exponential case (which corresponds to Erlang with $k=1$)
and the high-rate costs (i.e., when $\lambda \rightarrow \infty$) again approach 1 for all policies.
We note that
\revtwo{the Erlang distribution should approach that of the}{Erlang distribution results should approach those for}
deterministic when $k \rightarrow \infty$.}{It is straightforward to show that,
  for any $k \ge 1$, the cost ratios for each of the policies in the limiting cases of
  $\lambda \rightarrow 0$ and $\lambda \rightarrow \infty$ are the same as for exponentially distributed inter-request times.}

Results for
\revtwo{the deterministic distributions}{deterministic inter-request times}
are obtained by
\revfour{inserting equations (\ref{eqn:deterministic-A}) and (\ref{eqn:deterministic-B}) into}{using equations
  (\ref{eqn:deterministic-A}) and (\ref{eqn:deterministic-B}) to substitute into}
equations (\ref{eqn:C_k1_v2}), (\ref{eqn:C_k2always_v2}), (\ref{eqn:C_M_v2}), (\ref{eqn:C_k2_v2}),
taking limits (when needed), and simplifying the expressions
\revfour{to}{on}
a case-by-case basis.
\revfour{\revtwo{Again, the limits hold, with both}{We again observe the same asymptotes.  In the limits, both}
the {\em static baseline} and
\revtwo{different window-based}{window-based}
       {\em cache on $M^{th}$ request} policies, with $M \ge 2$,
       \revtwo{being}{performs}
       as good as {\em optimal offline}.  Only the {\em always on $M^{th}$} policies perform worse.
       Again, low-rate bounds (when $\lambda \rightarrow 0$) are $\frac{M+1}{M}$ for this policy.}{Again, the cost ratios
  for each of the policies in the limiting cases of $\lambda \rightarrow 0$ and $\lambda \rightarrow \infty$
  are the same as for exponentially distributed inter-request times.}
\revfour{The above bounds and observations are also clearly visible in Figure~\ref{fig:erlang},
  which highlights how the maximum cost difference for the}{Figure~\ref{fig:erlang} shows the cost ratio
  results for Erlang and deterministic inter-request times.  Note in particular how the peak cost ratio for}
        {\em single-window on $M^{th}$}, with $M \ge 2$,
        reduces as $k$ increases and inter-request times become increasingly deterministic (far-right sub figure).

\begin{figure*}[t]
  \centering
  \subfigure[Erlang, $k=2$]{
    \includegraphics[trim = 0mm 2mm 0mm 0mm, width=0.32\textwidth]{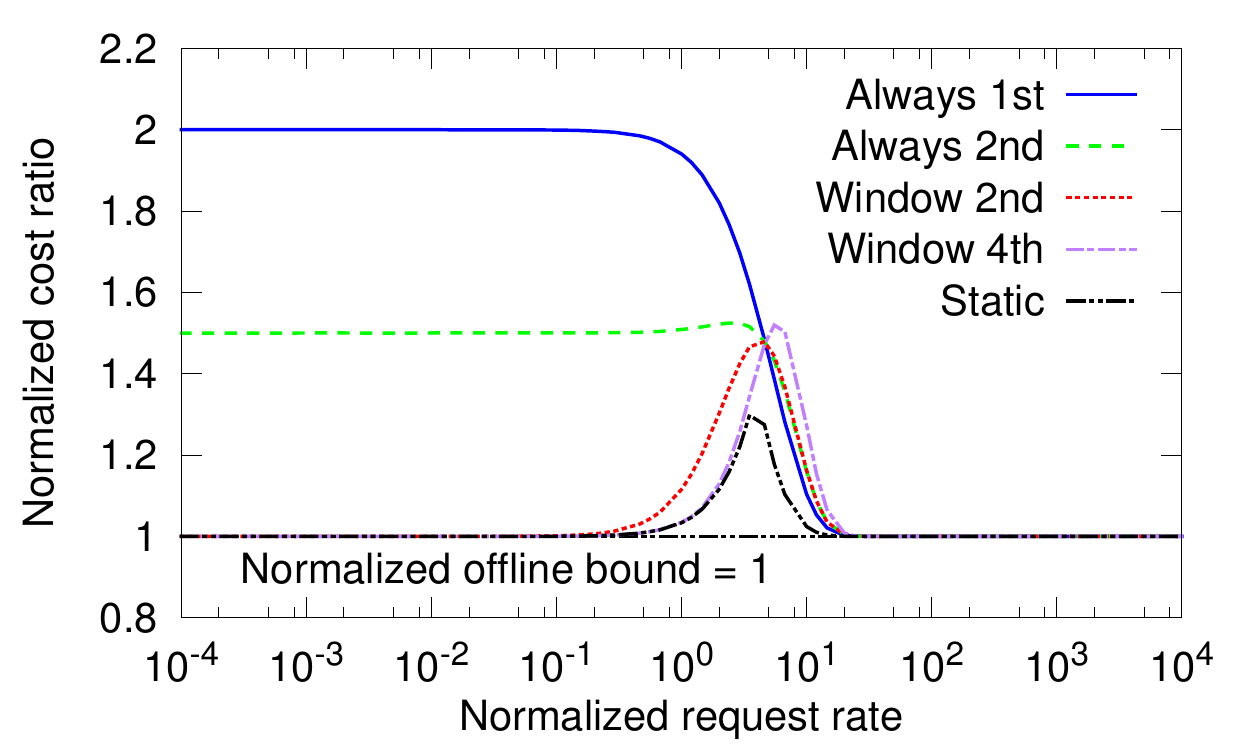}}
  \subfigure[Erlang, $k=4$]{
    \includegraphics[trim = 0mm 2mm 0mm 0mm, width=0.32\textwidth]{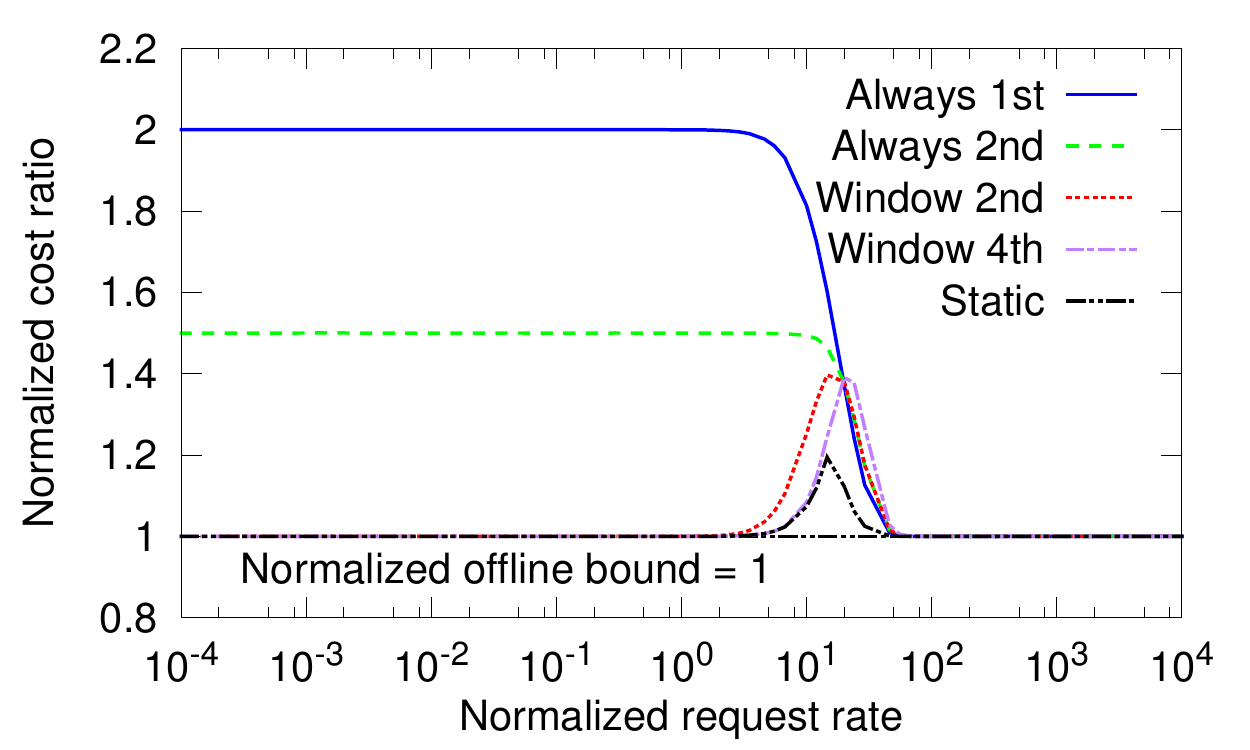}}
  \subfigure[Deterministic]{
    \includegraphics[trim = 0mm 2mm 0mm 0mm, width=0.32\textwidth]{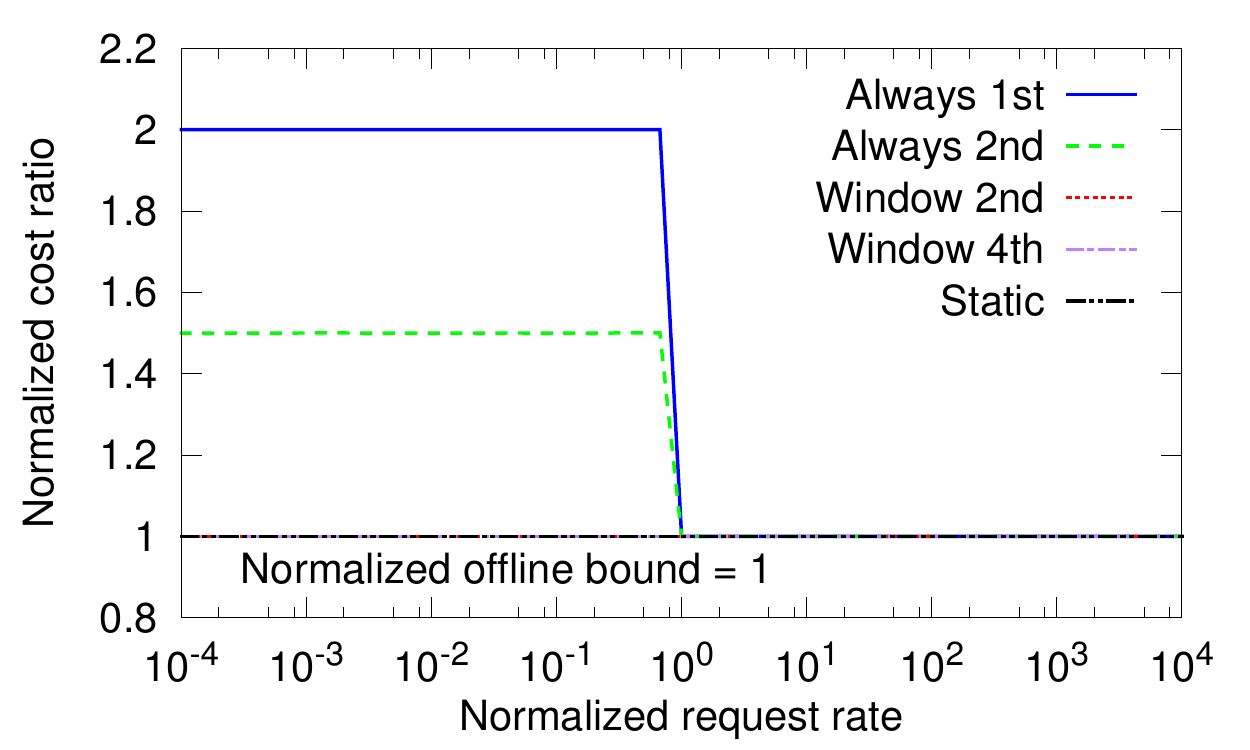}}
  \vspace{-8pt}
  \caption{\revfour{Summary results using lower-variability distributions and policies with $W=T=R$.}{Cost ratios for low variability inter-request time distributions and $W=T=R$.}}
  \label{fig:erlang}
  \vspace{-6pt}
\end{figure*}

{\bf Pareto:}
\revfour{To capture the impact under
heavy-tailed inter-request time distributions,
we used a Pareto distribution.  Here, equations
\revfour{(\ref{eqn:pareto-A}) and (\ref{eqn:pareto-B})}{(\ref{eqn:pareto-A1}), (\ref{eqn:pareto-A2}) and (\ref{eqn:pareto-B})}
are inserted into equations (\ref{eqn:C_k1_v2}), (\ref{eqn:C_k2always_v2}), (\ref{eqn:C_M_v2}), and (\ref{eqn:C_k2_v2}),
and some simplifications are performed to make the expressions more compact.
Figure~\ref{fig:pareto} then plot these functions for the same example policies as used in previous figures
for three different $\alpha$.}{Results for Pareto inter-request time distributions are obtained
  using equations (\ref{eqn:pareto-A1}), (\ref{eqn:pareto-A2}) and (\ref{eqn:pareto-B}) to substitute for $E[a_i]$, $F(t)$
  and the integral of $F(t)$ in
  equations (\ref{eqn:C_k1_v2}), (\ref{eqn:C_k2always_v2}), (\ref{eqn:C_M_v2}), and (\ref{eqn:C_k2_v2}).
  Figure~\ref{fig:pareto} shows cost ratio results for three different values of $\alpha$.}
We note that (as per Theorem~\ref{thm:static-pareto}),
{\em static baseline} performs very poorly when $\alpha \rightarrow 1$ (and $t_m$ is small).
This is illustrated by the large peak
\revfour{}{cost ratio}
in Figure~\ref{fig:pareto}(a), where $\alpha=1.1$.
For larger $\alpha$ (e.g., $\alpha=2$ in Figure~\ref{fig:pareto}(c)),
this peak reduces substantially.  Otherwise, the results are similar as for the other inter-request distributions
in that the maximum observed peaks are for {\em always on $1^{st}$},
and
\revtwo{that the}{in that}
       {\em single-window on $2^{nd}$} has a tighter bound than {\em single-window on $4^{th}$},
suggesting that  {\em single-window on $2^{nd}$} with $T=R$ is a good choice.

\begin{figure*}[t]
  \centering
  \subfigure[ $\alpha=1.1$]{
    \includegraphics[trim = 0mm 2mm 0mm 0mm, width=0.32\textwidth]{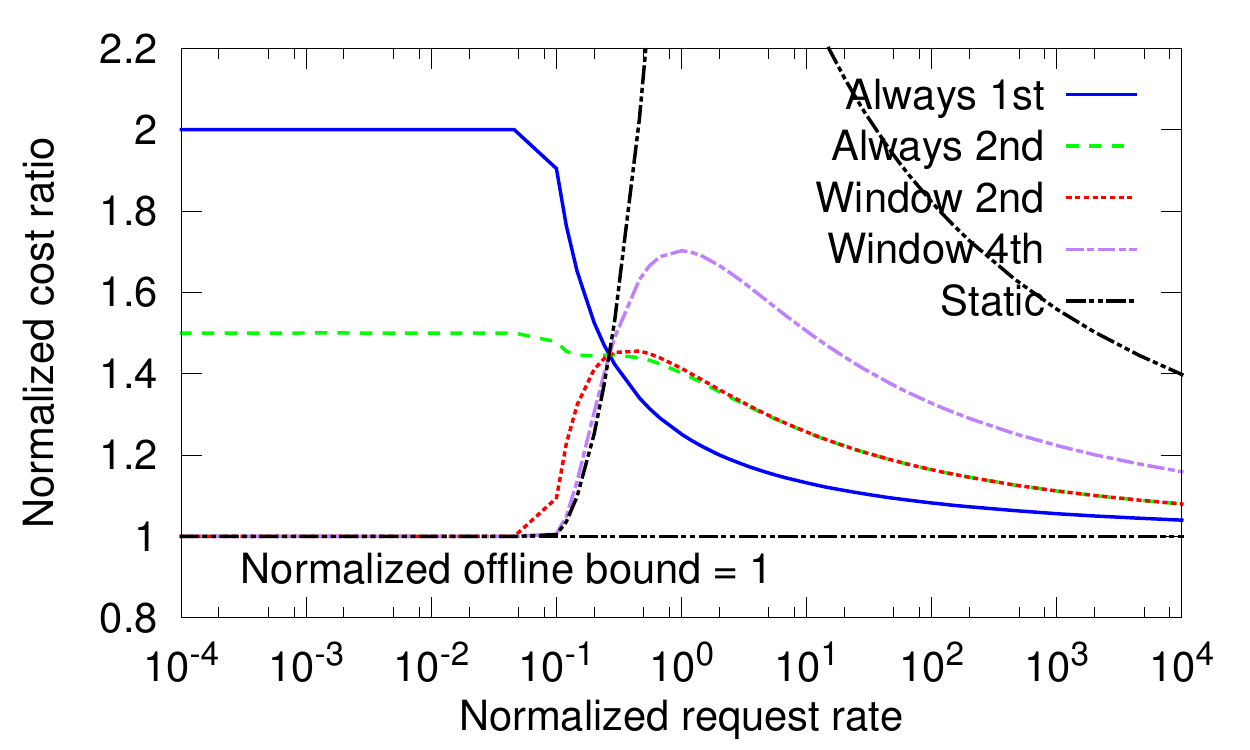}}
  \subfigure[ $\alpha=1.25$]{
    \includegraphics[trim = 0mm 2mm 0mm 0mm, width=0.32\textwidth]{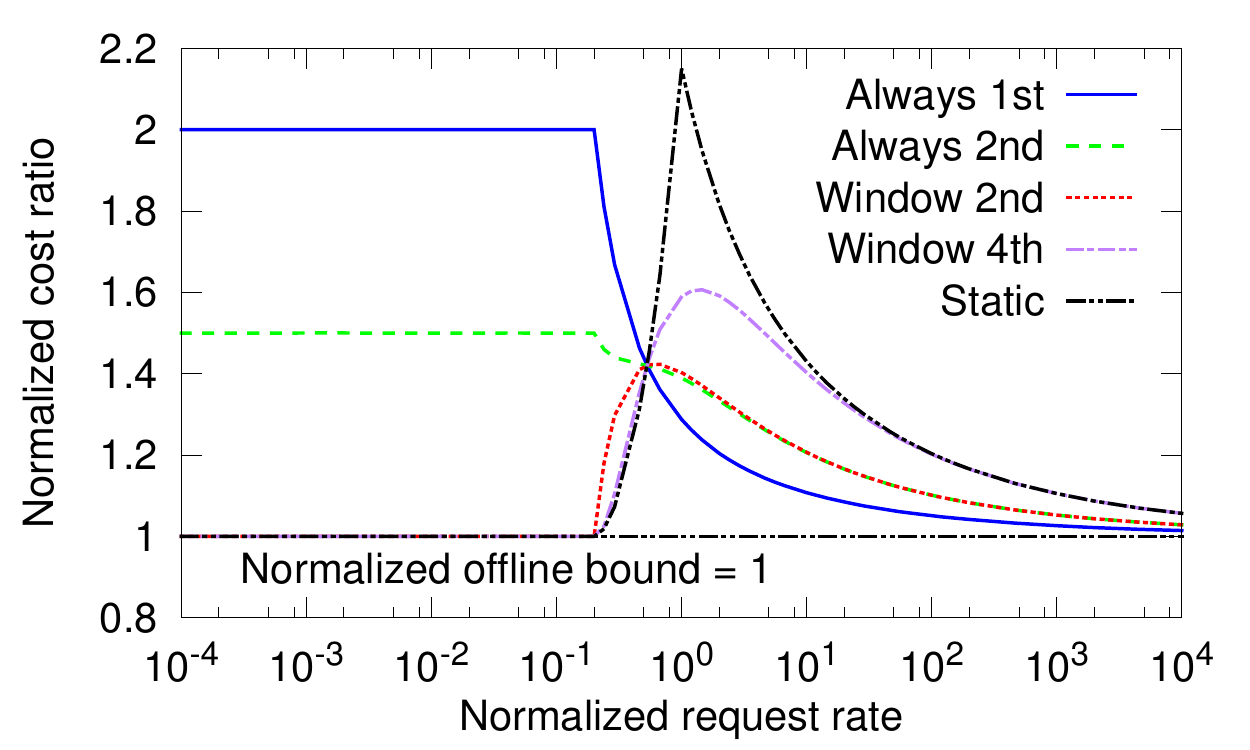}}
  \subfigure[ $\alpha=2$]{
    \includegraphics[trim = 0mm 2mm 0mm 0mm, width=0.32\textwidth]{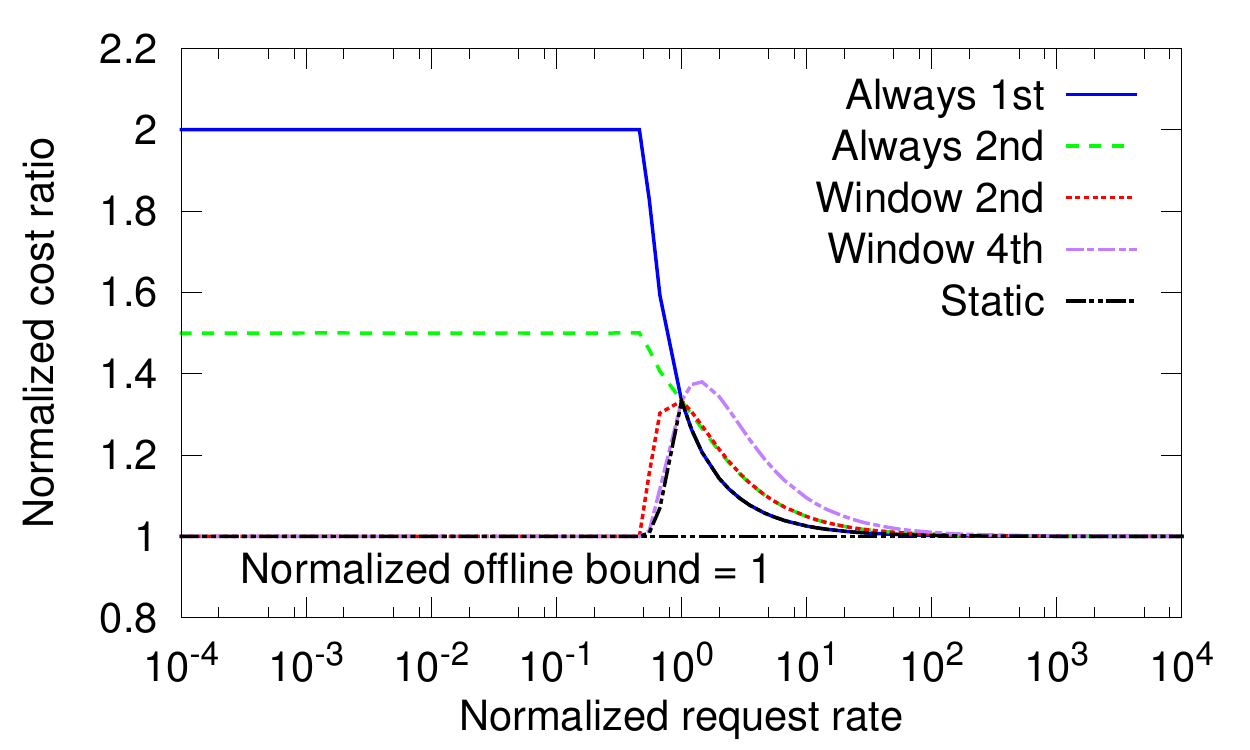}}
  \vspace{-8pt}
  \caption{\revfour{Heavy-tail summary results using Pareto distributions and policies with $W=T=R$.}{Cost ratios for Pareto inter-request time distributions and $W=T=R$.}}
  \label{fig:pareto}
  \vspace{-6pt}
 \end{figure*}

\begin{figure*}[t]
  \centering
  \subfigure[Pareto $\alpha=1.25$]{
    \includegraphics[trim = 0mm 2mm 0mm 0mm, width=0.32\textwidth]{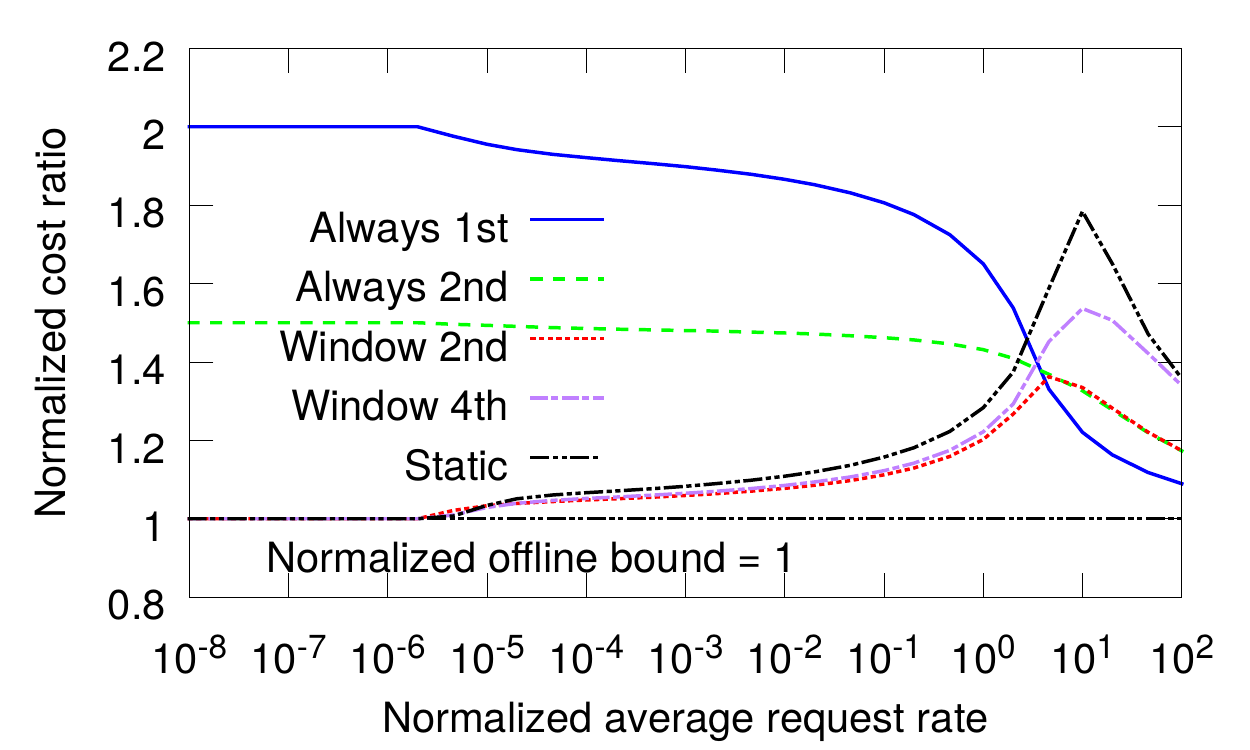}}
  \subfigure[Exponential]{
    \includegraphics[trim = 0mm 2mm 0mm 0mm, width=0.32\textwidth]{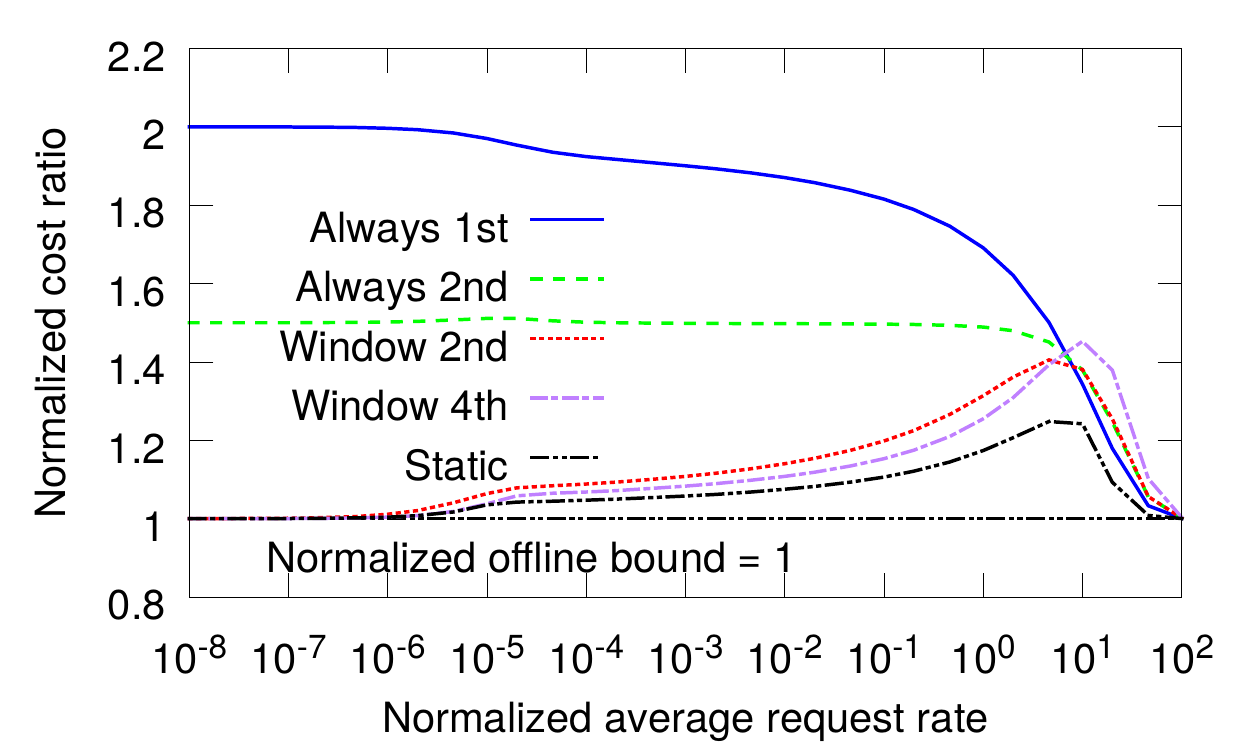}}
  \subfigure[Erlang $k=4$]{
    \includegraphics[trim = 0mm 2mm 0mm 0mm, width=0.32\textwidth]{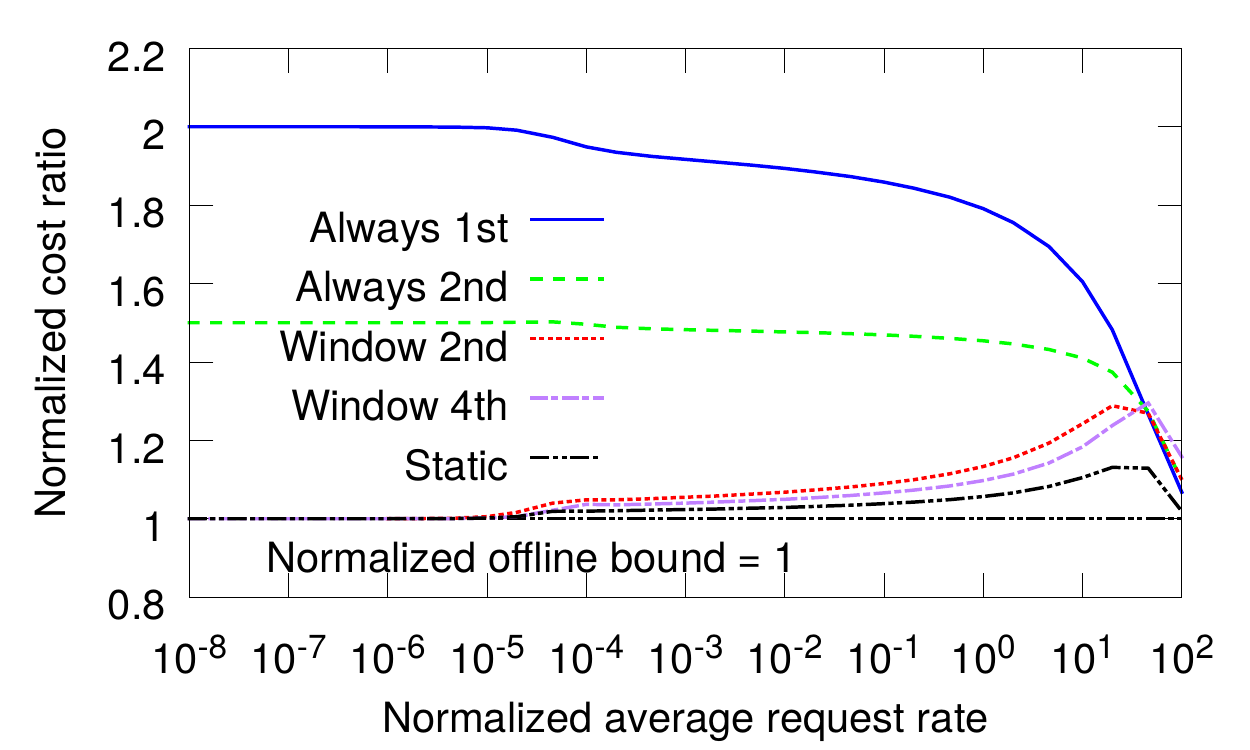}}
  \vspace{-8pt}
  \caption{Multi-file analysis for different inter-request \revfour{times;}{time distributions;} Zipf popularity distribution
    \revfour{($p_i \propto 1/i^{\gamma}$, $\gamma=1$, $N=1,000,000$).}{(frequency of requests to file $i$ proportional
      to $1/i^{\gamma}$, $\gamma=1$, and $1,000,000$ files).}}
  \label{fig:multi-base-case}
  \vspace{-6pt}
\end{figure*}

\section{Multi-file evaluation}\label{sec:multi}

Thus far we have focused primarily on
deriving analytic expressions and insights based on
the single file case.
In this section, we complement this analysis with
both analytic (Section~\ref{sec:multi-analytic}) and
trace-based (Section~\ref{sec:multi-trace})
evaluations for the multi-file case.

Throughout the section the different {\em cache on $M^{th}$ request} policies use the threshold values $W=T=R$.
Being the optimal worst-case choices,
\revfour{$W=T=R$ is the natural choice}{they are natural choices}
for this context,
since predicting individual object popularities is difficult and object
popularities in practice typically change over time.

\subsection{Heavy-tailed popularity analysis}\label{sec:multi-analytic}

File object popularities are typically highly skewed~\cite{GALM07,ZSGK09, MaSi15,CaEa17}.
For this analysis, we consider the delivery cost for a cache when the
file object popularity is Zipf
\revfour{distributed}{distributed with parameter $\gamma$} 
(i.e.,
\revfour{requests}{the frequency of requests}
to the $i^{\textrm{th}}$ most popular
\revfour{object}{file object}
is proportional to $\frac{1}{i^{\gamma}}$)
and all file objects have the same size.
Since both storage and bandwidth cost in our model scale proportional to the file size,
\revfour{extensions to account for variable sized objects are
  \revfour{trivial.}{trivial (e.g., by simply weighting each file or set of files by relative size).}}{results for variable-sized
  files could be easily obtained simply by weighting the costs for each file according to the file size.}
\revfour{Furthermore,
\revtwo{of}{for}
the same reason,
the results presented here are exact as long as the (expected)
average object size of each file popularity rank is the same.
If files of a particular popularity (e.g., the most popular files) are
larger/smaller than the other files of a particular workload,
\revtwo{this bias in object the popularity distribution (relating to file sizes)}{this bias}
can easily be adjusted for by assigning weights proportional to the object sizes.}{}

Figure~\ref{fig:multi-base-case} shows the cost ratio for the different policies
as a function of the normalized average request rate, when $\gamma$=$1$ and there are
\revfour{$N$=$10^6$ objects.}{$1,000,000$ files.}
To allow comparisons with the single-file case, we include results
\revtwo{for the cases that the individual inter-request times follow three
  different example distributions:}{for three
  \revfour{different distributions for the inter-request time of each file:}{forms for the inter-request time distribution of each file:}}
Pareto with $\alpha$=$1.25$ (Figure~\ref{fig:multi-base-case}(a)),
exponential (Figure~\ref{fig:multi-base-case}(b)), and
Erlang with $k$=$4$ (Figure~\ref{fig:multi-base-case}(c)).
\revfour{}{Different files have different distribution parameter values
  (value of $t_m$ for Pareto, $\lambda$ for exponential and Erlang)
  so as to achieve the desired Zipf request frequency distribution.}
\revtwo{The corresponding (omitted) results}{Results}
for Zipf popularity distributions
with $\gamma$=$0.75$ and $\gamma$=$1.25$ are very similar.

We note that {\em window on $M^{th}$} with $M=2$ has a
\revfour{maximum (peak)}{peak}
cost-ratio compared to the {\em offline optimal} of 1.4,
and significantly
\revfour{outperform}{outperforms}
the {\em always on $M^{th}$} policies.
These results again clearly highlight the value of a more selective insertion
\revtwo{policy, such as the window-based {\em cache on $M^{th}$ request} policy.}{policy.}

\revfour{Perhaps even more encouraging is}{Also important to note is}
the small gap between the
\revfour{{\em static}}{{\em static baseline}}
policy and the window-based policies
\revfour{when considering}{for}
exponential (Figure~\ref{fig:multi-base-case}(b)) and Erlang (Figure~\ref{fig:multi-base-case}(c))
distributed inter-request times, and
that the window-based policies outperform the
\revfour{{\em static}}{{\em static baseline}}
policy when inter-request times are Pareto distributed (Figure~\ref{fig:multi-base-case}(a)).
\revfour{Here, for the {\em static} policy, we have optimized whether to cache or not to cache each file individually
assuming perfect knowledge about both the request rate and the exact inter-request
\revtwo{distribution. For the case of exponential inter-request distribution (and as conjectured, also the Erlang distribution)
  this policy represents the
\revfour{optimal online TTL-policy}{{\em optimal online} policy}
  when any $T$ between zero and infinity can be selected
for each individual file separately.
Clearly, in practice, this knowledge (and hence also performance) is not possible to achieve.}{distribution,
  which yields minimum cost among all online policies for the distributions considered
  in Figures~\ref{fig:multi-base-case}(b) and~\ref{fig:multi-base-case}(c).}}{The {\em static baseline} policy optimizes
  its selection between always caching, and never caching, each file according to that file's inter-request time distribution.
  This yields minimum cost among all online policies for the distributions considered
  in Figures~\ref{fig:multi-base-case}(b) and~\ref{fig:multi-base-case}(c).}
Yet, {\em window on $2^{nd}$} and {\em window on $4^{th}$}
achieve close to this online bound, while treating all files the same.
\revfour{This is}{These results are}
highly encouraging and
\revfour{shows}{show}
that the same policy
can be used
\revfour{across}{for}
all files, regardless of popularity and
\revfour{inter-request pattern.}{the form of the inter-request time distribution.}

While the cost gap generally is small,
we note that the region over which the window-based
policies (and other online policies) leave a
\revfour{gap}{significant gap}
compared to the
\revfour{offline optimal}{{\em offline optimal}}
is substantially wider
for the multi-file case
\revfour{(say from $10^{-5}$ to $10^2$ for the exponential distribution; Figure~\ref{fig:multi-base-case}(b))
  than for the single file case (say $10^{-2}$ to 10; Figure~\ref{fig:exponential}).}{than for the single file case.
  For example, for exponential inter-request times, there is a significant gap in the multi-file case (Figure~\ref{fig:multi-base-case}(b))
  for normalized average request rate values from about $10^{-5}$ to $10^2$,
  while a significant gap in the single file case (Figure~\ref{fig:exponential}) appears only for request rate values from about $10^{-2}$ to $10$.}
\revfour{One reason for this is
\revtwo{that there is a high skew in the popularity distribution,
and that there almost always is a set of videos (somewhere in the popularity distribution)
with a non-negelectible request load that operates
in the region in which these policies leave a noticible gap for these videos.}{because of the highly-varying file popularities.
  Over a wide range of average request rates there are
  many
  files whose individual request rates fall in the region in which, in the single file case,
  there is a substantial gap compared to the {\em offline optimal}.}}{This is explained by the fact that in the multi-file case,
  files have widely-varying request rates, and over a wide range of average request rates there are files whose individual
  request rate falls in the region in which, in the single file case,
  there is a substantial gap compared to the {\em offline optimal}.}
Interestingly, the size of the set of files contributing to this gap will differ for different average request
\revfour{loads.}{rates.}
For example, at
\revfour{lower}{low}
average request rates,
there will be a small set
\revfour{(of relatively popular)}{of relatively popular}
files contributing to the gap.
However, due to the skew in popularity, this set
\revtwo{will be responsible for a
relatively larger per-file request load than the majority of the videos
(corresponding to the tail of less popular videos).}{will account for a disproportionate share of the total request volume.}
The small step around $10^{-6}$ to $10^{-5}$ is due to the most popular files entering this region. 
\revfour{At higher loads,}{At high average request rates,}
the number of
\revtwo{videos with a substantial individual gap
  (e.g., in peak-region of Figure~\ref{fig:exponential}) increases,}{files whose individual request rate
  falls in the region with a substantial gap increases,}
\revtwo{simultaneously as the relative fraction of all request that each of the individual videos operating in this region decreases,
  balancing the impact that this set of videos has on the overall gap.}{but these files now account for a
  disproportionately smaller share of the total request
  \revfour{load,}{rate,}
  an effect
  \revfour{which}{that}
  reduces the size of the peak gap for the {\em cache on $M^{th}$ request} policies.}
\revtwo{This effect also reduces the peak gap between the  {\em cache on $M^{th}$ request} policies
and the offline optimal.  Naturally,}{Note that}
for the {\em always on $M^{th}$} policies,
the worst-case gap (at low request rates) is the same as for the single file case.
However,
\revtwo{due to the skewed popularity distribution and many objects,
the top objects receive the bulk of the requests,
and the worst-case asymptotes are therefore not achieved until all videos reach this bound}{the worst-case asymptotes
  are not approached until the request rates for all files are low}
(which happens when the average request
\revfour{load}{rate}
falls somewhere between $10^{-4}$ and $10^{-6}$,
depending on the
\revfour{skew in the distributions considered).}{distribution skew).}

\subsection{Trace-based evaluation}\label{sec:multi-trace}

For our trace-based analysis,
we use a 20 month long trace capturing all YouTube video
requests from a campus network with 35,000 faculty, staff, and students.
The trace spans between July 1, 2008, and February 28, 2010,
and contains roughly 5.5 million requests to 2.4 million unique YouTube videos~\cite{CaEa17}.
This type of traffic is particularly interesting since file popularities
are ephemeral and there typically is a long tail of less popular files that individually
are viewed very few times, but that as an aggregate contribute to a significant part of the total views.
For example, in the university dataset,
90\% of the videos are requested three or fewer times,
and yet these videos make up half of the views observed on campus.

\begin{figure}[t]
  \centering
  \includegraphics[trim = 0mm 2mm 0mm 0mm, width=0.48\textwidth]{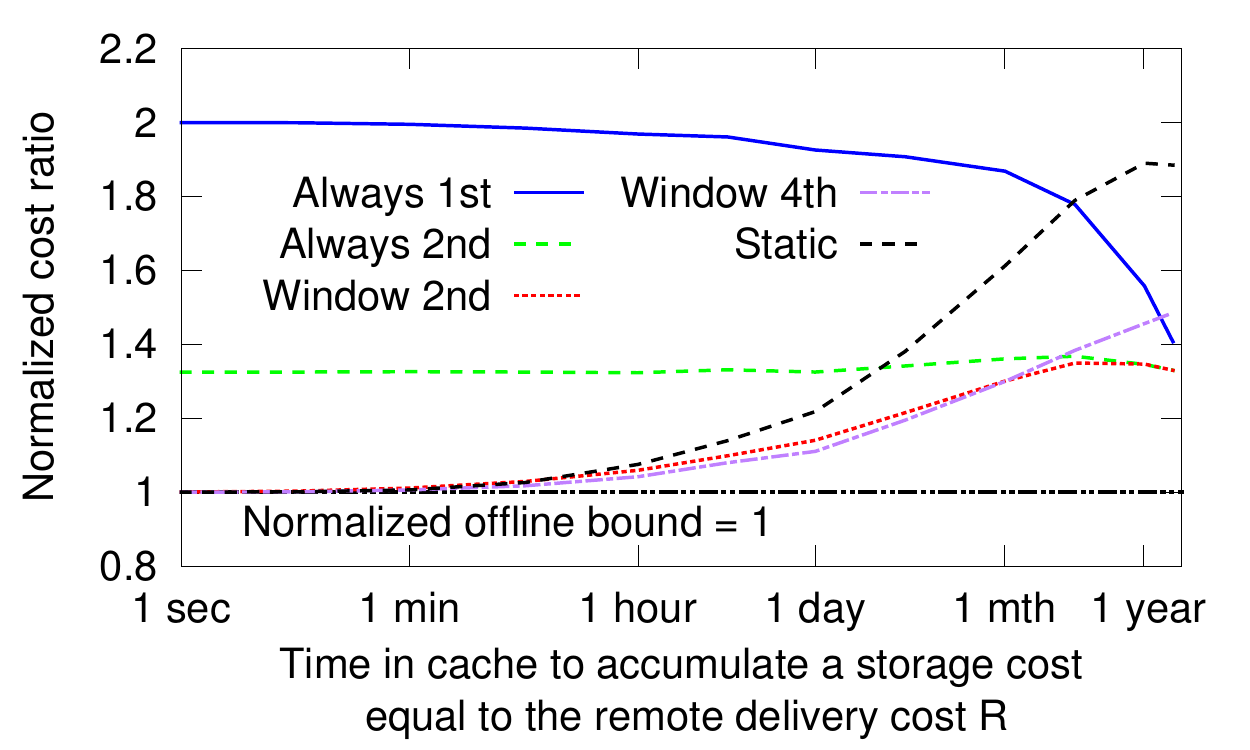}
  \vspace{-8pt}
  \caption{Trace-based simulation results with $W=T=R$.}
  \label{fig:uni-trace}
  \vspace{-6pt}
\end{figure}

Figure~\ref{fig:uni-trace} shows summary
\revfour{result}{results}
for our trace-based simulations.
Here,
\revfour{we plot}{for each policy we plot}
the ratio of the total aggregate delivery cost across all videos
divided by the corresponding delivery cost using the
\revfour{offline optimal}{{\em offline optimal}}
policy,
as a function of the time that a file would need to
be stored in cache to accumulate remote delivery cost $R$.
\revfour{}{With the unit normalization described in Section 2, $W = T = R$
  implies that storing a file in cache for $W = T$ time units would incur a cost equal to $R$,
  and so the x-axis values also correspond to the window sizes $W$ and $T$.}
For the {\em static baseline} policy, we make the optimistic assumptions that
(i) an {\em oracle} can be used to determine which of {\em always local} and {\em always remote}
will perform best for each individual video, and
(ii) in the case of {\em always local} the file object is not retrieved until the time of the first request (at a cost $R$).
In practice, such knowledge would not be available to any online policy.
Yet, the {\em window on $M^{th}$} policies significantly
outperform the {\em static baseline} policy.
This shows the importance of being selective in what
\revfour{is added and what is not}{is} added to the cache.

Due to the dominance of videos that see few
\revtwo{requests over the duration of the trace,}{requests,}
the results resemble
\revfour{the analytic results observed for lower request rates,}{the multi-file analytic results for lower average request rates,}
with the window-based {\em cache on $M^{th}$ request} policies performing the best.
For example, with 5.5 million requests to 2.4 million
\revfour{videos, the 20 month window corresponds to an {\em average} request rate of 2.3
  (in our normalized units used prior in the paper).}{videos over a 20 month period,
  a window size $W = T$ of 20 months would imply a normalized average request rate, as used on the x-axes in Figures 1-4, of 2.3.}
Furthermore, {\em window on $M^{th}$} with $M=4$ is a slightly better choice than $M=2$
for shorter than month-long caching thresholds $W=T=R$,
whereas for
\revfour{larger}{longer}
thresholds, $M=2$ is the better policy. 

\revfour{}{Much of the improvements over the {\em always on 1$^{st}$} policy, come from the
  {\em window on $M^{th}$} policies, with intermediate $M$, requiring smaller storage.
  For example, with a one-week threshold the average cache size at object evictions
  (across all object evictions) reduces from 153,729 objects (with {\em always on 1$^{st}$})
  to 57,652 ($M=2$) and 29,034 ($M=4$).  The corresponding values for a
  30-day  (``one month'' in Figure 5)
  threshold are: 343,139, 150,364 and 58,170.  Here, we also note that the variance
  in cache size needed over these time scales is relatively small,
  despite significant
  seasonal request volume variations in the trace
  (e.g., comparing summer breaks vs. regular term~\cite{CaEa17}).
  For example, in the case of the one-month threshold, the ratios of the maximum observed cache size
  to the minimum observed cache size at any two cache evictions instances (across the full 20-month trace)
  for these three policies are: 2.67, 2.21, and 2.02, respectively.}

To better understand
(i) which files contribute most of the absolute cost and
(ii) which files contribute most of the cost inflation (as seen in Figure~\ref{fig:uni-trace})
compared to the {\em offline optimal} bound, Figure~\ref{fig:uni-trace-breakdown} breaks down the cost
\revtwo{contributions made up for}{due to}
videos of different popularities.
Figure~\ref{fig:uni-trace-breakdown}(a) shows the costs of the different policies
associated with the videos with more than 20 views,
expressed relative to the total {\em offline optimal} bound cost.
This set contains 0.95\% of the unique videos and is responsible for 22.8\% of the views.
Figures~\ref{fig:uni-trace-breakdown}(b) and~\ref{fig:uni-trace-breakdown}(c) show the
corresponding results for the videos that have 4-20 views and 1-3 views over the duration of the 20-month long trace, respectively.
These two sets
\revfour{contains}{contain}
9.0\% and 90\% of the unique videos,
and are responsible for 27.6\% and 49.6\% of the views, respectively.

\begin{figure*}[t]
  \centering
  \subfigure[Top (more than 20 views)]{
    \includegraphics[trim = 0mm 2mm 0mm 0mm, width=0.32\textwidth]{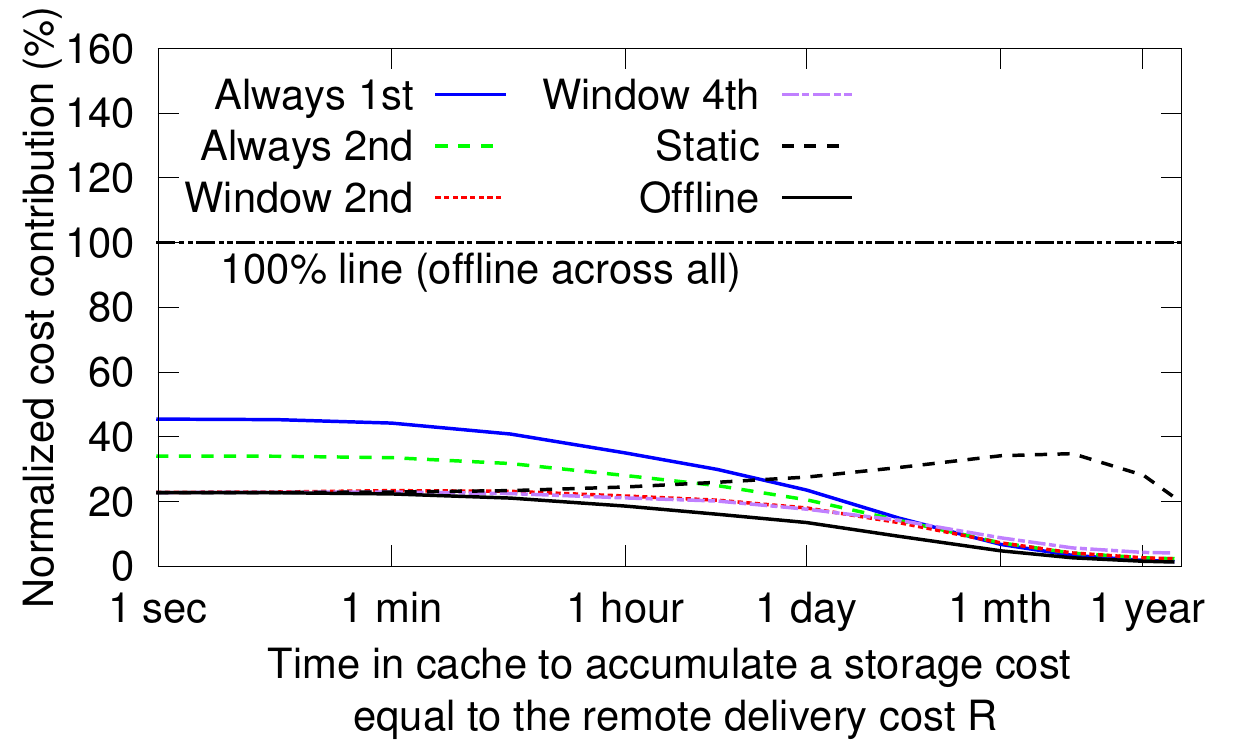}}
  \subfigure[Middle (4-20 views)]{
    \includegraphics[trim = 0mm 2mm 0mm 0mm, width=0.32\textwidth]{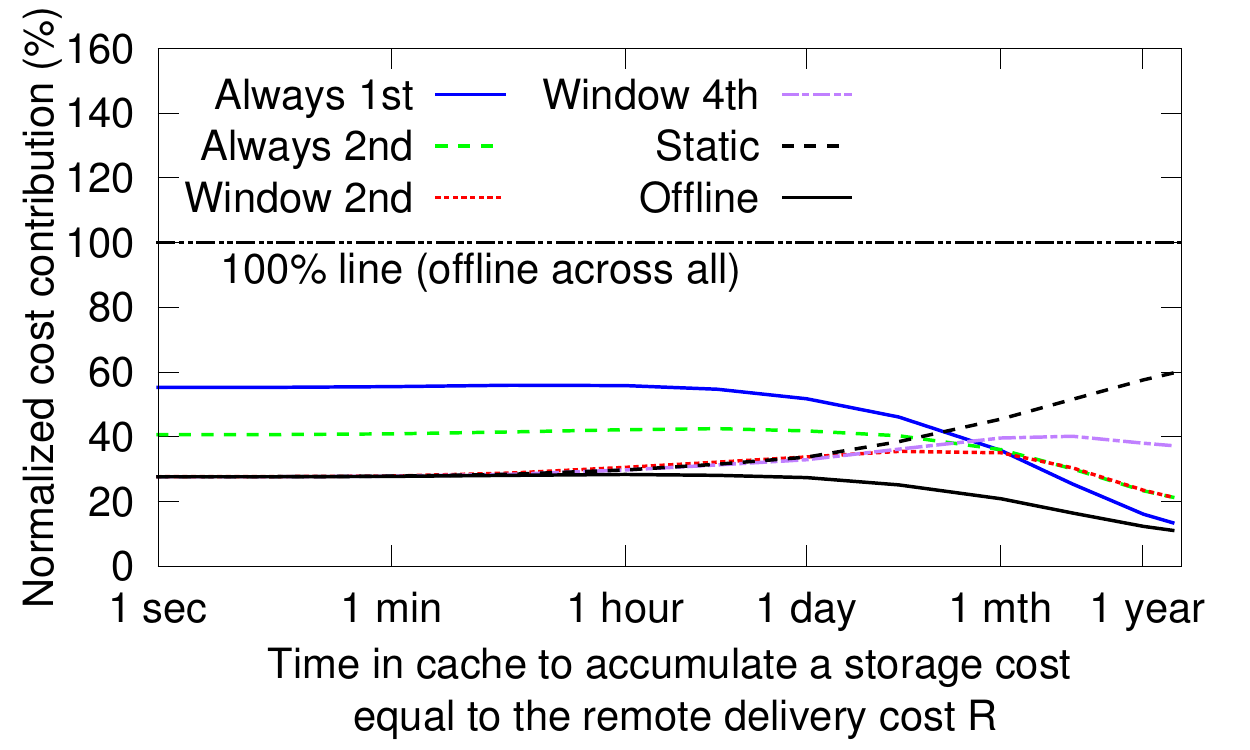}}
  \subfigure[Tail (1-3 views)]{
    \includegraphics[trim = 0mm 2mm 0mm 0mm, width=0.32\textwidth]{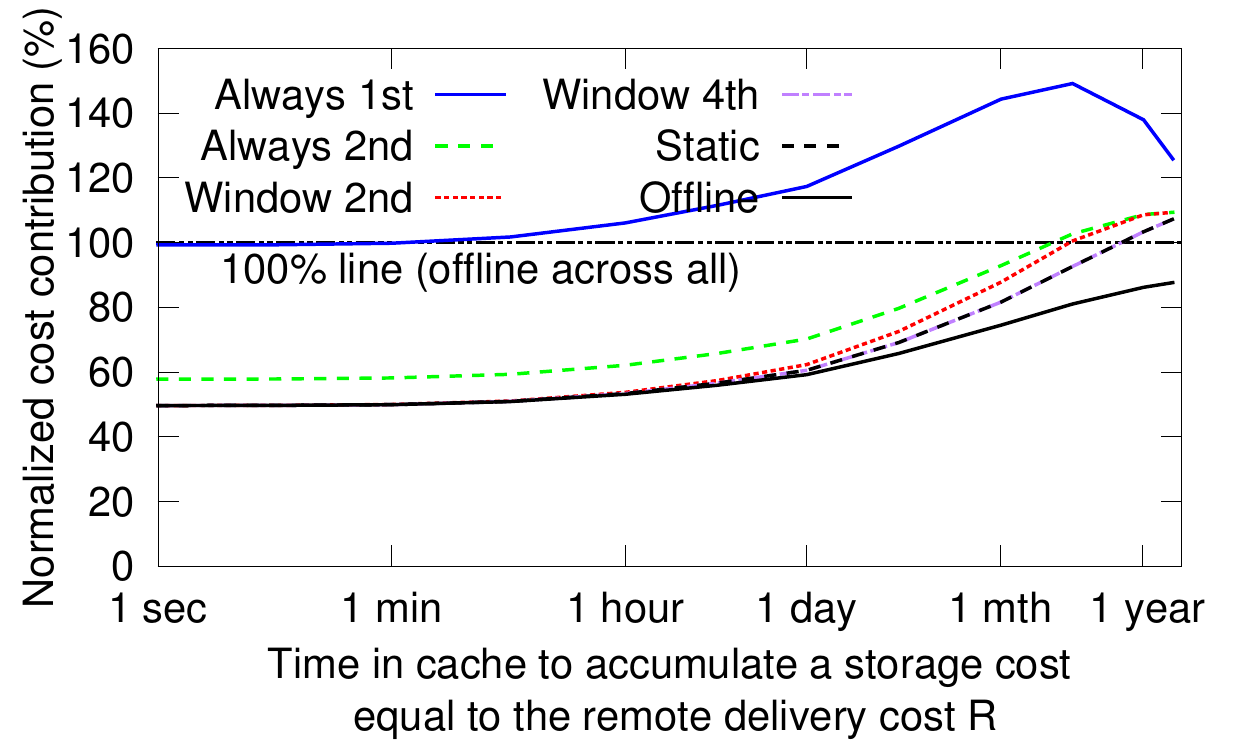}}
  \vspace{-8pt}
  \caption{Breakdown of cost contributions of the videos belonging to three different popularity categories. (University dataset.)}
  \label{fig:uni-trace-breakdown}
  \vspace{-6pt}
  \end{figure*}

These figures also show that the advantage of using window-based,
rather than purely counter-based, {\em cache on $M^{th}$ request} policies is consistent across
the three
\revtwo{classes of popularities,}{popularity classes,}
and that the fraction of the {\em offline optimal} caching cost
that the long-tail of less popular videos
\revtwo{}{contributes}
increases as the thresholds increase (and more videos are cached).
\revtwo{We note that much}{Much}
of the penalty of the {\em static baseline} policy is associated with the more popular
\revtwo{videos, as the relative cost penalty associated with these videos (curves for 4-20 and 20+ sets) sees relative increases
  for larger time $R$ values.  This increase is}{videos (comparing Figures~\ref{fig:uni-trace-breakdown}(a) and (b) to
  Figure~\ref{fig:uni-trace-breakdown}(c)), and longer thresholds,}
likely due to this policy
not capturing the ephemeral popularity of these videos.

Interestingly, even when the time in cache to accumulate
\revfour{}{a storage cost equal to the}
remote delivery cost $R$
is very small, the {\em few timers} (with 1-3 views) still contribute approximately 50\%
of the total cost for all policies,
except for {\em always on $1^{st}$} and {\em always on $2^{nd}$},
for which the contribution is even higher.
Overall, these results show the importance of
\revfour{discriminate}{selective}
caching policies such as the
window-based {\em cache on $M^{th}$} request policies analyzed in this paper.

\section{Related work}\label{sec:related}

Most existing caching works focus on replacement policies~\cite{PoBo03,BaOb00}.
However, recently it has been shown that the cache insertion policies play
a very important factor in reducing the total delivery costs~\cite{MaSi15,CaEa17}.
Motivated by these works,
this paper focuses on the delivery cost differences between
different
\revfour{discriminatory}{selective}
cache insertion policies.

Few papers (regardless of replacement policy) have modeled
\revfour{discriminatory}{selective}
cache insertion policies such as \emph{cache on $M^{th}$ request}.
This class of policies is motivated by the risk of cache pollution due to ephemeral content popularity and 
the long tail of one-timers (one-hit wonders) observed in edge networks~\cite{GALM07,ZSGK09, MaSi15,CaEa17}.
Recent works including trace-based evaluations of \emph{cache on $M^{th}$ request} policies~\cite{MaSi15,CaEa17}.
Carlsson and Eager~\cite{CaEa17} also present
simple analytic models for hit and insertion probabilities.
However, in contrast to the analysis presented here,
\revtwo{they ignore cache replacement, assuming}{they assume}
that content is not evicted until interest in the content has expired.
Garetto et al.~\cite{GaLM16, MaGL14} and Gast and Van Houdt~\cite{GaVa15,GaVa16} present
TTL-based recurrence expressions and
approximations for two variations of \emph{cache on $M^{th}$ request},
referred to as k-LRU and LRU(m) in their works.
However, none of these works present performance bounds or consider the total delivery cost.
In contrast, we derive both worst-case bounds and average-case analysis under a cost model
that captures both bandwidth and storage costs.

Finally,
it is important to note that
TTL-based
\revtwo{replacement}{eviction}
policies~\cite{JuBB03,BaMa05}
\revfour{}{(and variations thereof~\cite{CaEa18b})}
\revtwo{(considered in this paper) have}{have}
been found useful for approximating the performance
of capacity-driven replacement policies such as LRU~\cite{ChTW02, FrRR12, BDC+13, BGSC14, GaLM16}.
\revtwo{For an individual content provider,
  our results may therefore also be applicable to the case in which the provider use}{Our results may therefore also
  \revfour{be applicable to}{provide insight for}
  the case in which a content provider uses}
a fixed-sized cache.
Generalizations of the TTL-based Che-approximation~\cite{ChTW02} and TTL-based caches in general
have proven useful to analyze individual caches~\cite{ChTW02, FrRR12, BDC+13, BGSC14, GaLM16},
networks of caches~\cite{FNNT12, FDT+14, FNNT14, BGSC14, GaLM16},
and to optimize different system designs~\cite{CEGL14,DMT+16,FeRP16,MaTo15}.

\revrev{}{As we show here,
  \revfour{these type of elasticity}{elasticity}
  assumptions can also be a powerful toolbox
for deriving tight worst-case bounds and exact average-case cost ratios of different policies.
Furthermore,
\revfour{as argued in the paper,}{as discussed in Section 9.1,}
since both storage costs and bandwidth costs are proportional to the file sizes, the results
can also easily be extended to
\revfour{analyze variable}{scenarios with variable}
sized objects, at no additional computational cost.
In contrast, just finding lower and upper bounds for the cache miss rate of the
\revfour{optimal offline}{{\em optimal offline}}
policy is computationally expensive when caches are non-elastic~\cite{BeBH18} and even simple LRU
is hard to analyze under non-elastic constraints~\cite{King71,DaTo90}.}

\section{Conclusions}\label{sec:conclusion}

In this paper,
\revfour{we}{we consider the delivery costs of a content provider
that wants to minimize its delivery costs under the assumptions that the
resources it requires are elastic,
the content provider only pays for the resources that it
consumes, and costs are proportional to the resource usage.  Under these assumptions,
we}
first derived worst-case bounds for the optimal cost and competitive cost-ratios of different
classes of {\em cache on $M^{th}$ request} cache insertion policies.  Second, we derived explicit average cost
expressions and bounds under arbitrary inter-request
\revfour{}{time}
distributions, as well as for short-tailed inter-request
time distributions
\revtwo{(Erlang and deterministic), exponential inter-requests,}{(deterministic, Erlang, and exponential)}
and heavy-tailed inter-request
distributions (Pareto).  Finally, using these analytic results, we have
\revfour{present}{presented}
numerical evaluations and
cost comparisons that reveal insights into
\revtwo{the policies relative costs performance.}{the relative cost performance of the policies.}
Interestingly,
we have found that {\em single-window on $M^{th}$} with an intermediate $M$ (e.g., 2-4) and $T=R$ achieves
most of the benefits of this class of policies.  Choosing $T=R$ guarantees a worst-case competitive ratio
of $M+1$ (compared to the {\em optimal offline} policy), but typically performs much better.  For example,
we have found that
\revfour{}{this policy with}
$M=2$ closely tracks the
\revfour{{\em optimal online} when the optimal is known,}{{\em online optimal} policy for the short-tailed inter-request time distributions,}
and significantly
outperforms the standard
\revfour{non-discriminatory}{non-selective}
policy {\em always on $1^{st}$} across all inter-request distributions
considered here.  Using $M=4$ can result in further improvements for lower request rates (e.g., as associated with a
long tail of less popular file objects), but performs somewhat worse when request rates are intermediate
(where the gap between the online and offline policies is the greatest).  These results suggest that
{\em cache on $2^{nd}$} optimized to minimize worst-case costs provides good average performance,
making it an attractive choice for a wide range of practical conditions where request rates of
individual objects typically are not known and
\revfour{quickly can}{can quickly}
change.

\section*{Acknowledgements}

The edge-network trace was collected while the first author was a research associate at the University of Calgary.
We thank Carey Williamson and Martin Arlitt for providing access to this dataset.
This work was supported by funding from the Swedish Research Council (VR)
and the Natural Sciences and Engineering Research Council (NSERC) of Canada.

{\small
  	\bibliographystyle{ACM-Reference-Format}
\bibliography{references}
}

\appendix

\section{Additional worst-case proofs}

\subsection{Proof Theorem~\ref{thm:always-mth}: Always on $M^{th}$}

We next prove Theorem~\ref{thm:always-mth}, which specifies the worst-case properties of {\em always on $M^{th}$}.

\begin{proof}
  Case $T \le R$:
  For
  \revfour{$2 \le i \le |\mathcal{A}|$,}{$2 \le i \le N$,}
  let us define the following sets
  based on the operation of the {\em always on $M^{th}$} policy:
  $S^A_m = \{i | a_i \le R \land i$~is~the~$m^{th}$request$\}$,
  $S^C_m = \{i | R < a_i \land i$~is~the~$m^{th}$request$\}$,
  where we label request $i$ as the $m^{th}$ request when it is the $m^{th}$ request
  to the object since the object was removed from the cache most recently or
  the request sequence started.
  For the case that the previous request put the object into the cache
  or the
  \revfour{objects}{object}
  remained in the cache, we define the following sets:
  $S^A_+ = \{i | a_i \le T       \land i \notin \cup_{m=2}^M S^A_m \}$,
  $S^B_+ = \{i | T < a_i \le R \land i \notin \cup_{m=2}^M S^A_m \}$, and
  $S^C_+ = \{i | R < a_i     \land i \notin \cup_{m=2}^M S^C_m \}$.
  Note that the set $S^A_+$ corresponds to cache hits using the {\em always on $M^{th}$} policy,
  and that sets $S^B_+$ and $S^C_+$ corresponds to cases where the counter is reset after the object
  has been removed from the cache (and the cache
  \revfour{endured}{incurred}
  an extra storage cost $T$ after most recent prior request).

  Now, for an arbitrary
  \revrev{arrival pattern}{request sequence}
  $\mathcal{A}$, we can bound the cost of the {\em always on $M^{th}$} policy
  as follows:
  $C^{always}_{M,T} \le R + \sum_{m=2}^M |S^A_m| R + \sum_{m=2}^M |S^C_m| R + \sum_{i \in S^A_+} a_i + (|S^B_+|+|S^C_+|)(R+T) + T$, where the final $T$ only
  is needed if the last request in the sequence is from set $S^A_+$.
  (In all other cases the bound becomes loose.)
  Now, noting that (i) $|S^A_M|+|S^C_M| \le |S^A_{M-1}|+|S^C_{M-1}| \le ... \le |S^A_{2}|+|S^C_{2}|$,
  (ii) $\sum_{m=2}^M (|S^A_m|+S^C_m|) \le (M_1) (|S^A_{2}|+|S^C_{2}|)$, and
  (iii) $|S^B_{+}| + |S^C_{+}| \le |S^A_{2}|+|S^C_{2}|$,
  we can write $C^{always}_{M,T} \le R + M (|S^A_{2}|+|S^C_{2}|) R + (|S^A_{2}|+|S^C_{2}|) T + \sum_{i \in S^A_+} a_i + T$.

  For the {\em optimal offline} policy, we note that all requests in the sets
  $S^A_m$, $S^A_+$, $S^B_+$
  \revfour{corresponds}{correspond}
  to cache hits (associated with an extra storage cost $a_i$),
  whereas the remaining requests are cache misses (associated with a remote access cost $R$).
  Therefore, for the same
  \revrev{arrival pattern}{request sequence}
  $\mathcal{A}$, the cost of the optimal (offline) policy can be bounded as follows:
  $C_{opt}^{offline} = R + \sum_{i \in {\cup_{m=2}^M S^A_m}} a_i + \sum_{m=2}^M |S^C_m| R + \sum_{i \in S^A_+} a_i + \sum_{i \in S^B_+} a_i + |S^C_+|R$
  $\ge R + \sum_{m=2}^M R |S^C_m| + T |S^B_+| + |S^C_+| R + \sum_{i \in S^A_+} a_i$
  $\ge R + T (|S^A_2| + |S^C_2|) + \sum_{i \in S^A_+} a_i$.
  Here, we have used that
  (i) $ \sum_{i \in {\cup_{m=2}^M S^A_m}} a_i \ge 0$,
  (ii) $\sum_{i \in S^B_+} a_i \ge T |S^B_+|$,
  (iii) $\sum_{i \in S^C_+} a_i \ge R |S^C_+|$, and
  (vi) $T |S^B_+| + |S^C_+| R \ge T (|S^B_+| + |S^C_+|) = T (|S^A_2| + |S^A_2|$).
  Taking the ratio
  \begin{align}
    \frac{C^{always}_{M,T}}{C_{opt}^{offline}} & \le \frac{R + M (|S^A_{2}|+|S^C_{2}|) R + (|S^A_{2}|+|S^C_{2}|) T + \sum_{i \in S^A_+} a_i + T}{R + T (|S^A_2| + |S^C_2|) + \sum_{i \in S^A_+} a_i} \nonumber\\
    & \le \frac{R + M (|S^A_{2}|+|S^C_{2}|) R + (|S^A_{2}|+|S^C_{2}|) T + T}{R + T (|S^A_2| + |S^C_2|)},
  \end{align}
  it is easy to show that the worst case scenario happens with $|S^A_{2}|+|S^C_2| \rightarrow \infty$
  and that the worst-case bound is minimized by setting $T=R$.
  (To see this, note that $\frac{d}{dx}(\frac{R+MRx+Tx+T}{R+Tx}) = \frac{MR^2-T^2}{(R+Tx)^2} \ge 0$.)
  In this case the worst-case ratio reduces to $(M+1)$.

  Finally, we show that this ratio is achievable by a request pattern in which
  requests occurs in batches of $M$ requests,
  \revfour{where each such batch is}{with consecutive batches}
  spaced by more than $R$ time units.
  In this case,
  we have $a_i=0$ for all $i \in S^A_m$, $|S^A_+|=|S^B_+|=|S^C_+|=|S^B_m|=0$ for all $m$,
  and $|S^A_m|=|S^C_+|$ for all $m$.  In each batch cycle,
  the {\em always on $M^{th}$} policy downloads the object $M$ times
  from the server and keeps it in
  \revfour{stored}{the cache}
  for $R$ time units (at a total cost of $(M+1)R$ per batch).
  In contrast, the {\em optimal offline} policy downloads a single copy (at cost $R$),
  serves all $M$ requests using this copy, and then instantaneously deletes the copy (to avoid storage costs).

  Case $R \le T$:
  Let us define the following sets for
  \revfour{$2 \le i \le |\mathcal{A}|$:}{$2 \le i \le N$:}
  $G^A_m = \{i | a_i \le R \land i$~is~the~$m^{th}$request$\}$,
  $G^C_m = \{i | R < a_i \land i$~is~the~$m^{th}$request$\}$,
  where $2 \le m \le M$, and
  $G^A_+ = \{i | a_i \le R       \land i \notin \cup_{m=2}^M G^A_m \}$,
  $G^B_+ = \{i | R < a_i \le T \land i \notin \cup_{m=2}^M G^C_m \}$, and
  $G^C_+ = \{i | T < a_i     \land i \notin \cup_{m=2}^M G^C_m \}$.
  With these sets,
  \revfour{only}{only the requests in}
  sets $G^A_+$ and $G^B_+$ correspond to cache hits (with associated cost $a_i$)
  with the {\em always on $M^{th}$} policy.  Furthermore, with this policy,
  the requests in set $G^C_+$ corresponds to cases where the counter is reset after the object has been removed from the cache.
  These cache misses are therefore associated with an extra storage cost $T$
  (corresponding to the time the object was in the cache
  without being requested again after the most recent earlier request).
  Now, for
  an arbitrary
  \revrev{arrival pattern}{request sequence}
  $\mathcal{A}$, we can bound the cost of this policy as follows:
  $C^{always}_{M,T} \le R + \sum_{m=2}^M |G^A_m| R + \sum_{m=2}^M |G^C_m| R + \sum_{i \in G^A_+} a_i + \sum_{i \in G^B_+} a_i + |G^C_+|(R+T) + T$.
  Now, noting that (i) $\sum_{m=2}^M (|G^A_m| + |G^C_m|) \le (M-1) (|G^A_{2}| + |G^C_{2}|)$,
  (ii) $|G^C_{+}| = |G^A_{2}| + |G^C_{2}|$,
  (iii) $\sum_{i \in G^B_{+}} a_i \le |G^B_{+}|T$,
  we can write $C^{always}_{M,T} \le R + M |G^C_{+}| R + (|G^C_{+}|) T + |G^B_{+}|T + \sum_{i \in G^A_+} a_i +T$.
  Similarly, for the same
  \revrev{arrival pattern}{request sequence}
  $\mathcal{A}$, the cost of the {\em optimal offline} policy can be bounded as follows:
  $C_{opt}^{offline} = R + \sum_{m=2}^M \sum_{i \in G^A_m} a_i + \sum_{m=2}^M |G^C_m| R + \sum_{i \in G^A_+} a_i + |G^B_+| R + |G^C_+| R$
  $ \ge R + |G^B_+| R + |G^C_+| R + \sum_{i \in G^A_+} a_i$,
  where we have used that
  (i) $\sum_{i \in G^A_m} a_i \ge 0$, and
  (ii) $|G^C_m| \ge 0$.
  Taking the ratio
  \begin{align}
    \frac{C^{always}_{M,T}}{C_{opt}^{offline}} & \le \frac{  R + M |G^C_{+}| R + |G^C_{+}| T + |G^B_{+}|T + \sum_{i \in G^A_+} a_i +T}{R + |G^B_+| R + |G^C_+| R + \sum_{i \in G^A_+} a_i}\nonumber\\
    & \le \frac{R + M |G^C_{+}| R + |G^C_{+}| T + |G^B_{+}|T + T}{R + |G^B_+| R + |G^C_+| R}
  \end{align}
  \revfour{As earlier,}{it can be seen that, as earlier,}
  this ratio is minimized when $T=R$,
  for which it is bounded by $(M+1)$ when $|G^B_{+}|=0$ and $(|G^C_{+}|) \rightarrow \infty$.
  (To see this, note that $\frac{d}{dx}(\frac{R+MRx+Tx+BT+T}{R+BR+Rx}) = \frac{BM+M-1}{(B+x+1)^2} \ge 0$.)
  It is trivial to see that the same request pattern (but with batches separated by
  \revfour{at least}{more than}
  $T$ rather than $R$)
  results in the worst case being achieved.  This shows that the bound is tight.
\end{proof}

\subsection{Proof Theorem~\ref{thm:cache-2nd-two}: Single-window on $M^{th}$}

We next prove Theorem~\ref{thm:cache-2nd-two}, which specifies the worst-case properties of {\em single-window on $M^{th}$}.

\begin{proof}
  Case $T \le R$:
  For
  \revfour{$2 \le i \le |\mathcal{A}|$,}{$2 \le i \le N$,}
  let us define the following sets
  based on the operation of the {\em single-window on $M^{th}$} policy:
  $S^A_m = \{i | a_i \le T \land i$~is~an~$m^{th}$~candidate$\}$,
  $S^B_m = \{i | T < a_i \le R \land i$~is~an~$m^{th}$~candidate$\}$, and
  $S^C_m = \{i | R < a_i \land i$~is~an~$m^{th}$~candidate$\}$,
  where we say that a request is an $m^{th}$ candidate
  whenever the previous request in the request sequence to the object
  set the counter to $(m-1)$.
  Note that the first overall request and the first request after the object
  has been removed from the cache
  always sets the counter to one (and
  \revfour{the object}{the next request to the object}
  hence becomes a $2^{nd}$ candidate).
  For the case that the previous request put the object into the cache
  or the object remained in the cache, we define the following sets:
  $S^A_+ = \{i | a_i \le T       \land i \notin \cup_{m=2}^M S^A_m \}$,
  $S^B_+ = \{i | T < a_i \le R \land i \notin \cup_{m=2}^M S^B_m \}$, and
  $S^C_+ = \{i | R < a_i     \land i \notin \cup_{m=2}^M S^C_m \}$.
  Note that
  \revfour{the}{the requests in the}
  set $S^A_+$ corresponds to cache hits using the {\em single-window on $M^{th}$} policy,
  and that sets $S^B_+$ and $S^C_+$ correspond to cases where the counter is reset after the object has been removed from the cache
  (and the cache
  \revfour{endured}{incurred}
  an extra storage cost $T$ after the most recent prior request).

  Now, for
  \revrev{any}{an}
  arbitrary
  \revrev{arrival pattern}{request sequence}
  $\mathcal{A}$, we can bound cost of the {\em single-window on $M^{th}$} policy
  as follows:
  $C^{window}_{M,T} \le R + \sum_{m=2}^M |S^A_m| R + \sum_{m=2}^M |S^B_m| R + \sum_{m=2}^M |S^C_m| R + \sum_{i \in S^A_+} a_i + (|S^B_+|+|S^C_+|)(R+T) + T$,
  where the final $T$ only is needed if the last request in the sequence is from set $S^A_+$.
  (In all other cases the bound becomes loose.)
  Now, noting that (i) $|S^A_M| \le |S^A_{M-1}| \le ... \le |S^A_{2}|$,
  (ii) $\sum_{m=2}^M |S^A_m| \le (M_1) |S^A_{2}|$,
  (iii) $\sum_{m=2} |S^B_{m}| + \sum_{m=2} |S^C_{m}| + |S^B_{+}| + |S^C_{+}| \le |S^A_{2}|$, and
  (iv) $|S^B_{+}| + |S^C_{+}| \le |S^A_{2}|$,
  we can write $C^{window}_{M,T} \le R + M |S^A_{2}| R + |S^A_{2}| T + \sum_{i \in S^A_+} a_i +T$.

  For the {\em optimal offline} policy, we note that all requests in the sets
  $S^A_m$, $S^B_m$, $S^A_+$, $S^B_+$ correspond to cache hits (associated with an extra storage cost $a_i$),
  whereas the remaining requests are cache misses (associated with a remote access cost $R$).
  Therefore, for the same
  \revrev{arrival pattern}{request sequence}
  $\mathcal{A}$, the cost of the optimal (offline) policy can be bounded as follows:
  $C_{opt}^{offline} = R + \sum_{i \in {\cup_{m=2}^M S^A_m}} a_i + \sum_{i \in {\cup_{m=2}^M S^B_m}} a_i + \sum_{m=2}^M |S^C_m| R + \sum_{i \in S^A_+} a_i + \sum_{i \in S^B_+} + |S^C_+|R \ge R + T (\sum_{m=2}^M |S^B_m| + |S^B_+|) +  R (\sum_{m=2}^M |S^C_m| + |S^C_+|) + \sum_{i \in S^A_+} a_i \ge R + |S^A_2| T + \sum_{i \in S^A_+} a_i$.
  Here, we have used that
  (i) $ \sum_{i \in {\cup_{m=2}^M S^A_m}} a_i \ge 0$,
  (ii) $\sum_{i \in S^B_m} a_i \ge T |S^B_m|$,
  (iii) $\sum_{i \in S^C_m} a_i \ge T |S^C_m|$,
  (iv) $\sum_{i \in S^B_+} a_i \ge T |S^B_+|$,
  (v) $\sum_{i \in S^C_+} a_i \ge T |S^C_+|$,
  (vi) $T\sum_{m=2}^M |S^B_m| + R\sum_{m=2}^M |S^C_m| + T|S^B_+| + R|S^C_+| \ge T (\sum_{m=2}^M |S^B_m| + \sum_{m=2}^M |S^C_m| + |S^B_+| + |S^C_+|) = T |S^A_{2}|$.
  Taking the ratio
  \begin{align}
    \frac{C^{window}_{M,T}}{C_{opt}^{offline}} & \le \frac{R+M|S^A_{2}|R + |S^A_2|T + \sum_{i \in S^A_+} a_i + T}{R+T |S^A_{2}|+\sum_{i \in S^A_+} a_i} \nonumber\\
    & \le \frac{R+M|S^A_{2}|R + |S^A_2|T + T}{R+T |S^A_{2}|}
  \end{align}
  it is easy to show that the worst case scenario happens with $|S^A_{2}| \rightarrow \infty$
  and that the worst-case bound is minimized by setting $T=R$.
  In this case the worst-case ratio reduces to $(M+1)$.

  Finally, we show that this ratio is achievable by a request pattern in which
  requests occurs in batches of $M$ requests,
  \revfour{where each such batch is}{with consecutive batches}
  spaced by more than $R$ time units.
  In this case,
  we have $a_i=0$ for all $i \in S^A_m$, $|S^A_+|=|S^B_+|=|S^C_+|=|S^B_m|=0$ for all $m$,
  and $|S^A_m|=|S^C_+|$ for all $m$.  In each batch cycle,
  the {\em single-window on $M^{th}$} policy downloads the object $M$ times
  from the server and keeps it in
  \revfour{stored}{the cache}
  for $R$ time units (at a total cost of $(M+1)R$ per batch).
  In contrast, the {\em optimal offline} policy downloads a single copy (at cost $R$),
  serves all $M$ requests using this copy, and then instantaneously deletes the copy (to avoid storage costs).

  Case $R \le T$:
  Let us define the following sets for
  \revfour{$2 \le i \le |\mathcal{A}|$:}{$2 \le i \le N$:}
  $G^A_m = \{i | a_i \le R \land i$~is~an~$m^{th}$~candidate$\}$,
  $G^B_m = \{i | R < a_i \le T \land i$~is~an~$m^{th}$~candidate$\}$,
  $G^C_m = \{i | T < a_i \land i$~is~an~$m^{th}$~candidate$\}$,
  where $2 \le m \le M$, and
  $G^A_+ = \{i | a_i \le R       \land i \notin \cup_{m=2}^M G^A_m \}$,
  $G^B_+ = \{i | R < a_i \le T \land i \notin \cup_{m=2}^M G^B_m \}$, and
  $G^C_+ = \{i | T < a_i     \land i \notin \cup_{m=2}^M G^C_m \}$.
  With these sets,
  \revfour{only}{only the requests in the}
  sets $G^A_+$ and $G^B_+$ correspond to cache hits (with associated cost $a_i$)
  with the {\em single-window on $M^{th}$} policy.  Furthermore, with this policy,
  the requests in set $G^C_+$ correspond to cases where the counter is reset after the object has been removed from the cache.
  These cache misses are therefore associated with an extra storage cost $T$ (corresponding to the time the object was in the cache
  without being requested again after the most recent earlier request).
  Now, for
  \revrev{any}{an}
  arbitrary
  \revrev{arrival pattern}{request sequence}
  $\mathcal{A}$, we can bound the cost of this policy as follows:
  $C^{window}_{M,T} \le R + \sum_{m=2}^M |G^A_m| R + \sum_{m=2}^M |G^B_m| R + \sum_{m=2}^M |G^C_m| R + \sum_{i \in G^A_+} a_i + \sum_{i \in G^B_+} a_i + |G^C_+|(R+T) + T$.
  Now, noting that (i) $\sum_{m=2}^M (|G^A_m| + |G^B_m|) \le (M_1) (|G^A_{2}| + |G^B_{2}|)$,
  (ii) $\sum_{m=2}^M |G^C_{m}| + |G^C_{+}| = |G^A_{2}| + |G^B_{2}|$,
  (iii) $|G^C_{+}| = |G^A_{2}| + |G^B_{2}| - \sum_{m=2}^M |G^C_{m}| \le |G^A_{2}| + |G^B_{2}|$, and
  (iv) $\sum_{i \in G^B_{+}} a_i \le |G^B_{+}|T$,
  we can write $C^{window}_{M,T} \le R + M (|G^A_{2}|+|G^B_{2}|) R + (|G^A_{2}|+|G^B_{2}|) T + |G^B_{+}|T + \sum_{i \in S^A_+} a_i +T$.
  Similarly, for the same
  \revtwo{arrival}{request}
  pattern $\mathcal{A}$, the cost of the {\em optimal offline} policy can be bounded as follows:
  $C_{opt}^{offline} = R + \sum_{m=2}^M \sum_{i \in G^A_m} a_i + \sum_{m=2}^M |G^B_m| R + \sum_{m=2}^M |G^C_m| R + \sum_{i \in G^A_+} a_i + |G^B_+| R + |G^C_+| R \le R + (|G^A_{2}|+|G^B_{2}|) R + |G^B_{+}|T + \sum_{i \in S^A_+} a_i$,
  where we have used that
  (i) $\sum_{i \in G^A_m} a_i \ge 0$,
  (ii) $|G^B_m| \le 0$, and
  (iii) $\sum_{m=2}^M |G^C_{m}| + |G^C_{+}| = |G^A_{2}| + |G^B_{2}|$.
  Taking the ratio
  \begin{align}
    \frac{C^{window}_{M,T}}{C_{opt}^{offline}} & \le {\textstyle{ \frac{R + M (|G^A_{2}|+|G^B_{2}|) R + (|G^A_{2}|+|G^B_{2}|) T + |G^B_{+}|T + \sum_{i \in S^A_+} a_i +T}{R + (|G^A_{2}|+|G^B_{2}|) R + |G^B_{+}|T + \sum_{i \in S^A_+} a_i}}} \nonumber\\
    & \le \frac{R + M (|G^A_{2}|+|G^B_{2}|) R + (|G^A_{2}|+|G^B_{2}|) T + |G^B_{+}|T + T}{R + (|G^A_{2}|+|G^B_{2}|) R + |G^B_{+}|T}.
  \end{align}
  \revfour{As earlier,}{it is easy to show that, as earlier,}
  this ratio is minimized when $T=R$,
  for which it is bounded by $(M+1)$ when $|G^B_{+}|=0$ and $(|G^A_{2}|+|G^B_{2}|) \rightarrow \infty$.
  It is trivial to see that the same request pattern (but with batches separated by
  \revfour{at least}{more than}
  $T$ rather than $R$)
  results in the worst case being achieved.  This shows that the bound is tight.

\end{proof}

\subsection{Proof Theorem~\ref{thm:cache-Mth}: Dual-window on $2^{nd}$}

We next prove Theorem~\ref{thm:cache-Mth}, which specifies the worst-case properties of {\em dual-window on $2^{nd}$}.

\begin{proof}
  Case $W \le T \le R$:
  Let us define the following sets for
  \revfour{$2 \le i \le |\mathcal{A}|$:}{$2 \le i \le N$:}
  $S^A_2 = \{i | a_i < W \land i~\textrm{is~a~2}^{nd}~\textrm{candidate} \}$,
  $S^B_2 = \{i | W \le a_i < T \land i~\textrm{is~a~2}^{nd}~\textrm{candidate} \}$,
  $S^C_2 = \{i | T \le a_i < R\land i~\textrm{is~a~2}^{nd}~\textrm{candidate} \}$,
  $S^D_2 = \{i | R \le a_i \land i~\textrm{is~a~2}^{nd}~\textrm{candidate} \}$,
  $S^A_+ = \{i | a_i < W \land i \notin S^A_2 \}$,
  $B^B_+ = \{i | W \le a_i < T \land i \notin S^B_2 \}$,
  $S^C_+ = \{i | T \le a_i < R \land i \notin S^C_2 \}$, and
  $S^D_+ = \{i | R \le a_i \land i \notin S^D_2 \}$.
  We can now write $C^{M=2}_{W,T} = R + (|S^A_2|+|S^B_2|+|S^C_2|+|S^D_2|) R + \sum_{i \in S^A_+} a_i + \sum_{i \in S^B_+} a_i + (|S^C_+|+|S^D_+|)(R+T)$,
  and $C_{opt} = R + \sum_{i \in S^A_2} a_i +  \sum_{i \in S^B_2} a_i + \sum_{i \in S^C_2} a_i + |S^D_2| R +  \sum_{i \in S^A_+} a_i +  \sum_{i \in S^B_+} a_i + \sum_{i \in S^C_+} a_i + |S^D_+| R$.
  Now, noting that $|S^C_+|+|S^D_+|=|S^A_2|$, and making similar simplifications
  as in prior proofs, it is easy to show that:
  \begin{align}
    \frac{C^{window}_{M=2,W,T}}{C_{opt}^{offline}} \le \frac{R + 2|S^A_2|R + |S^A_2|T + |S^B_2|R + |S^C_2|R}{R + |S^A_2|T + |S^B_2|W + |S^C_2|T}.
  \end{align}
  This expression is minimized when $W \rightarrow T$ and $T \rightarrow R$.
  With these choices,
  the worst-case bound of 3 is achieved when $|S^B_2|=|S^C_2|=0$ and $|S^A_2| \rightarrow \infty$
  (and the same worst-case
  \revrev{arrival pattern}{request sequence}
  as used for the single parameter version).

  Case $W \le R \le T$:
  Let us define the following sets for
  \revfour{$2 \le i \le |\mathcal{A}|$:}{$2 \le i \le N$:}
  $H^A_2 = \{i | a_i < W \land i~\textrm{is~a~2}^{nd}~\textrm{candidate} \}$,
  $H^B_2 = \{i | W \le a_i < R \land i~\textrm{is~a~2}^{nd}~\textrm{candidate} \}$,
  $H^C_2 = \{i | R \le a_i < T\land i~\textrm{is~a~2}^{nd}~\textrm{candidate} \}$,
  $H^D_2 = \{i | T \le a_i \land i~\textrm{is~a~2}^{nd}~\textrm{candidate} \}$,
  $H^A_+ = \{i | a_i < W \land i \notin H^A_2 \}$,
  $H^B_+ = \{i | W \le a_i < R \land i \notin H^B_2 \}$,
  $H^C_+ = \{i | R \le a_i < T \land i \notin H^C_2 \}$, and
  $H^D_+ = \{i | T \le a_i \land i \notin H^D_2 \}$.
  We can now write $C^{M=2}_{W,T} = R + (|H^A_2|+|H^B_2|+|H^C_2|+|H^D_2|) R + \sum_{i \in H^A_+} a_i + \sum_{i \in H^B_+} a_i + \sum_{i \in H^C_+} a_i + |H^D_+|(R+T)$,
  and $C_{opt} = R + \sum_{i \in H^A_2} a_i +  \sum_{i \in H^B_2} a_i + (|H^C_2|+|H^D_2|) R +  \sum_{i \in H^A_+} a_i +  \sum_{i \in H^B_+}\ a_i + (|H^C_+|+|H^D_+|) R$.
  Now, noting
  that $|H^D_+|=|H^A_+|$, and making similar simplifications
  as in prior proofs, it is easy to show that:
  \begin{align}
    \frac{C^{window}_{M=2,W,T}}{C_{opt}} \le \frac{R + 2|H^A_2|R + |H^A_2|T + |H^B_2|R}{R + |H^A_2|R + |H^B_2|W}.
  \end{align}
  This expression is minimized when $W \rightarrow T$ and $T \rightarrow R$.
  With these choices,
  the worst-case bound of 3 is achieved when $|H^B_2|=0$ and $|H^A_2| \rightarrow \infty$.

  Case $R \le W \le T$:
  Let us define the following sets for
  \revfour{$2 \le i \le |\mathcal{A}|$:}{$2 \le i \le N$:}
  $G^A_2 = \{i | a_i < R \land i~\textrm{is~a~2}^{nd}~\textrm{candidate} \}$,
  $G^B_2 = \{i | R \le a_i < W \land i~\textrm{is~a~2}^{nd}~\textrm{candidate} \}$,
  $G^C_2 = \{i | W \le a_i < T\land i~\textrm{is~a~2}^{nd}~\textrm{candidate} \}$,
  $G^D_2 = \{i | T \le a_i \land i~\textrm{is~a~2}^{nd}~\textrm{candidate} \}$,
  $G^B_+ = \{i | a_i < R \land i \notin G^A_2 \}$,
  $G^A_+ = \{i | R \le a_i < W \land i \notin G^B_2 \}$,
  $G^C_+ = \{i | W \le a_i < T \land i \notin G^C_2 \}$, and
  $G^D_+ = \{i | T \le a_i \land i \notin G^D_2 \}$.
  We can now write $C^{window}_{M=2W,T} = R + (|G^A_2|+|G^B_2|+|G^C_2|+|G^D_2|) R + \sum_{i \in G^A_+} a_i + \sum_{i \in G^B_+} a_i + \sum_{i \in G^C_+} a_i + |G^D_+|(R+T)$,
  and $C_{opt}^{offline} = R + \sum_{i \in G^A_2} a_i + (|G^B_2|+|G^C_2|+|G^D_2|) R + \sum_{i \in G^A_+} a_i +  (|G^B_+| + |G^C_+| + |G^D_+|) R$.
  Now, noting that $|G^D_+|=|G^A_2|+|G^B_2|$, and making similar simplifications
  as in prior proofs, it is easy to show that:
  \begin{align}
    \frac{C^{window}_{M=2,W,T}}{C_{opt}^{offline}} \le \frac{R + 2|G^D_+|R + (|G^B_+|+|G^C_+|+|G^D_+|)T}{R + (|G^B_+|+|G^C_+|+|G^D_+|)R}.
  \end{align}
  This expression is minimized when $T \rightarrow R$ (and $W=T=R$).
  With these choices,
  the worst-case bound of 3 is achieved when $|G^B_+|=|G^C_+|=0$ and $|G^D_+| \rightarrow \infty$.

\end{proof}

\end{document}